\documentclass[runningheads]{llncs}

\usepackage{underscore}

\usepackage{verbatim}

\usepackage{cite}

\usepackage{amsmath}
\usepackage{mathtools} 
\usepackage{amssymb}
\usepackage{bm}

\usepackage{hyperref}  

\usepackage{tikz}
\usetikzlibrary{arrows,automata,positioning}
\usepackage{pgfplots}
\pgfplotsset{compat=1.14}

\usepackage{algorithm}
\usepackage[noend]{algpseudocode}
\algrenewcommand\algorithmicindent{1.5em}%

\algblockdefx[Match]{Match}{EndMatch}%
   [1]{\textbf{match} #1}%
   {\textbf{EndMatch}}
\algblockdefx[Match]{Case}{EndCase}%
   [1]{#1 $\Rightarrow$}%
   {\textbf{EndCase}}
\algtext*{EndMatch}%
\algtext*{EndCase}%

\usepackage{proof}

\newcommand{\ldot}{\mathpunct{.}}
\newcommand{\set}[1]{\{ #1 \}}

\newcommand{\tuple}[1]{{\langle #1 \rangle}}
\newcommand{\dom}{\mathrm{dom}}

\newcommand{\bool}{\mathbb{B}}

\newcommand{\overlineindex}[1]{\overline{#1\mkern-3mu}\mkern 3mu}

\newcommand{\dep}{\mathit{dep}}
\newcommand{\exdep}{\mathit{exdep}}
\newcommand{\var}{\mathit{var}}
\newcommand{\poset}{\mathit{poset}}

\newcommand{\forkext}{\mathbf{FEx}}

\newcommand{\unsat}{\mathit{unsat}}
\newcommand{\sat}{\mathit{sat}}
\newcommand{\assignment}{\alpha}
\newcommand{\assignments}{\mathcal{A}}
\newcommand{\prj}{\mathit{prj}}
\newcommand{\bound}{\mathit{bound}}

\renewcommand{\matrix}{\varphi}
\renewcommand{\models}{\vDash}
\newcommand{\nmodels}{\nvDash}

\newcommand{\bfalse}{\bot}
\newcommand{\btrue}{\top}
\newcommand{\bundef}{\mathit{undef}}

\newcommand{\asscupdot}{\mathbin{\dot\sqcup}}
\newcommand{\subassign}{\sqsubseteq}

\newcommand{\thesistodo}[1]{} 

\newcommand{\qres}{Q\text{-resolution}}

\newcommand{\nexptime}{\textsc{NExpTime}}

\newcommand{\dcaqe}{\textsc{dCAQE}}
\newcommand{\caqe}{\textsc{CAQE}}
\newcommand{\idq}{\textsc{iDQ}}
\newcommand{\hqs}{\textsc{HQS}}
\newcommand{\hqspre}{\textsc{HQSPre}}
\newcommand{\iprover}{\textsc{iProver}}

\makeatletter
\newcommand\smallparagraph{\@startsection{paragraph}{4}{\z@}%
                       {4\p@ \@plus 0\p@ \@minus 0\p@}%
                       {-0.5em \@plus -0.22em \@minus -0.1em}%
                       {\normalfont\normalsize\bfseries}}
\makeatother

\title{Clausal Abstraction for DQBF (full version)}
\author{%
Leander Tentrup\inst{1}
\and
Markus N.~Rabe\inst{2}
}
\authorrunning{L. Tentrup, M.N. Rabe}
\institute{
Reactive Systems Group, Saarland University, Saarbr\"ucken, Germany
\and
Google Research, Mountain View, California\footnote[1]{Work partially done while at University of California, Berkeley.}
}

\begin{document}

\maketitle

\begin{abstract}
Dependency quantified Boolean formulas (DQBF) is a logic admitting existential quantification over Boolean functions, which allows us to elegantly state synthesis problems in verification such as the search for invariants, programs, or winning regions of games.
In this paper, we lift the clausal abstraction algorithm for quantified Boolean formulas (QBF) to DQBF.
Clausal abstraction for QBF is an abstraction refinement algorithm that operates on a sequence of abstractions that represent the different quantifier levels. 
For DQBF we need to generalize this principle to partial orders of abstractions.
The two challenges to overcome are: (1) Clauses may contain literals with incomparable dependencies, which we address by the recently proposed proof rule called Fork Extension, and (2) existential variables may have spurious dependencies, which we prevent by tracking consistency requirements during the execution.
Our implementation $\dcaqe$ solves significantly more formulas than the existing DQBF algorithms.
\end{abstract}

\section{Introduction}  \label{sec:introduction}

The search for functions given declarative specifications is often called the synthesis problem and it is considered to be an extremely hard algorithmic problem.
The synthesis of invariants, programs, or winning regions of (finite) games can all be expressed as the existence of a function $f \colon \mathbb{B}^m \to \mathbb{B}^n$ such that for all tuples of inputs $x_1,\dots,x_k \in \mathbb{B}^m$ some relation $\varphi(x_1,f(x_1),\dots,x_k,f(x_k))$ over function applications of $f$ is satisfied.
While it is possible to specify these problems in SMT or in first-order logic, existing algorithms struggle to solve even simple instances of synthesis queries.

In order to develop a new algorithmic approach for synthesis problems, we focus on the simplest logic admitting the existential quantification over Boolean functions, dependency quantified Boolean formulas~(DQBF).
However, existing algorithms for DQBF perform poorly, in particular on synthesis problems~\cite{conf/tacas/FaymonvilleFRT17}.
This is not surprising:
Typical synthesis queries contain two or more function applications, i.e. are of the form $\exists f\ldot\forall x_1,x_2\ldot \varphi(x_1,f(x_1),x_2,f(x_2))$, and involve bit-vector variables, e.g. $x_1,x_2\in \mathbb{B}^n$.
The so far best performing algorithm for DQBF needs to expand either $x_1$ or $x_2$ in order to reach a linear quantifier prefix, which can then be converted to a QBF~\cite{conf/date/GitinaWRSSB15}.
This means that they often reduce to QBF formulas that are exponential in $n$.

Abstraction refinement algorithms have been very successful for QBF, winning the recent editions of the annual QBF competition~\cite{journals/ai/JanotaKMC16,conf/ijcai/JanotaM15,conf/fmcad/RabeT15,journals/ai/PulinaS19}.
Inspired by this success story, we lift the abstraction refinement algorithm called \emph{clausal abstraction}~\cite{conf/fmcad/RabeT15} to DQBF.
The idea of clausal abstraction for QBF is to split the given quantified problem into a sequence of propositional problems, one for each quantifier in the quantifier prefix, and instantiate a SAT solver for each of them.
The SAT solvers solve the quantified problem by communicating assignments (representing examples and counter-examples) to their neighbors.

Lifting clausal abstraction to DQBF comes with two major challenges:
First, clausal abstraction is based on $\qres$~\cite{conf/cav/Tentrup17} and $\qres$ is sound but incomplete for DQBF~\cite{journals/tcs/BalabanovCJ14}.
In particular, clauses may contain variables from incomparable quantifiers, so called information forks~\cite{conf/sat/Rabe17} which characterizes the reason for incompleteness.
We address this problem using the \emph{Fork Extension}~\cite{conf/sat/Rabe17} proof rule, which allows us to split clauses with information fork into a set of clauses without information fork by introducing new variables.
Second, clausal abstraction relies on the linear quantifier order of QBFs in prenex normal form.
For DQBF, however, quantifiers can form an arbitrary partial order.
When building a linear order by over-approximating the dependencies of existential variables and applying clausal abstraction naively, those variables may have \emph{spurious dependencies}, i.e. they may only be able to satisfy all the constraints, if they depend on variables that are not allowed by the Henkin quantifiers.
We show how to record consistency requirements, i.e., partial Skolem functions, that guarantee that existential variables solely depend on their stated dependencies.

In this paper, we present the first abstraction based solving approach for DQBF.
The algorithm successfully applies recent insight in solving quantified Boolean formulas: It is based on the versatile and award-winning clausal abstraction framework~\cite{conf/fmcad/RabeT15,conf/ijcai/JanotaM15,journals/corr/Tentrup16,conf/sat/Tentrup16,conf/cav/Tentrup17} and leverages progress in DQBF proof systems~\cite{conf/sat/Rabe17}.
Their integration in this work is non-trivial.
To handle the non-linear dependencies, we use an over-approximation of the dependencies together with consistency requirements.
Further, we turn clausal abstraction into an incremental algorithm that can accept new clauses and variables during solving.
Our experiments show that our approach consistently outperforms first-order reasoning~\cite{conf/pos/FrohlichKB12} on the DQBF benchmarks and it is especially well-suited for the synthesis benchmark set~\cite{conf/tacas/FaymonvilleFRT17} where expansion-based solvers fall short.

\section{Preliminaries}  \label{sec:preliminaries}

Let $\mathcal{V}$ be a finite set of propositional variables.
We use the convention to denote universally quantified variables (short also \emph{universals}) by $x$ and existentially quantified variables (or \emph{existentials}) by $y$. 
The set of all universals is denoted $\mathcal{X}$, and the set of all existentials is denoted $\mathcal{Y}$.
For sets of universals and existentials we use $X$ and $Y$, respectively.
We consider DQBF of the form
$\forall x_1 \ldot \dots \forall x_n \ldot \exists y_1(H_1) \ldot \dots \exists y_m(H_m) \ldot \varphi$,
that is, DQBF begin with universal quantifiers followed by \emph{Henkin quantifiers} and the quantifier-free part $\varphi$.
A Henkin quantifier $\exists y(H)$ introduces a new variable $y$, like a normal quantifier, but also specifies a set $H \subseteq \mathcal{X}$ of \emph{dependencies}.
A \emph{literal} $l$ is either a variable $v \in \mathcal{V}$ or its negation $\overline{v}$.
We call the disjunction $C = (l_1 \lor l_2 \dots \lor l_n)$ over literals a \emph{clause}, and assume w.l.o.g. that the propositional part of DQBFs are given as a conjunction of clauses, i.e., in conjunctive normal form~(CNF). 
We call the propositional part $\varphi$ of a DQBF in CNF the \emph{matrix} and we use $C_i$ to refer to the $i$th clause of $\varphi$ where unambiguous. 
For convenience, we treat clauses also as a set of literals and we treat matrices as a set of clauses and use the usual set operations for their manipulation.
We denote by $\var(l)$ the variable $v$ corresponding to literal $l$.
For literals $l$ of existential variables with dependency set $H$ we define $\dep(l)=H$.
For literals of universal variables we define $\dep(l) = \{\var(l)\}$.
We lift the operator $\dep$ to clauses by defining $\dep(C) = \bigcup_{l \in C} \dep(l)$.
We define $C|_V$ for some clause $C$ and set of variables $V$ as the clause $\set{l \in C \mid \var(l) \in V}$.

Given a set of variables $V \subseteq \mathcal{V}$, an \emph{assignment} of $V$ is a function $\assignment : V \rightarrow \mathbb B$ that maps each variable $v \in V$ to either true ($\top$) or false ($\bot$).
A \emph{partial assignment} is a partial function from $V$ to $\mathbb B$, i.e. it may be \emph{undefined} on some inputs. 
To improve readability, we represent (partial) assignments also as a conjunction of literals (i.e., a cube), e.g., we write $x_1 \overlineindex{x_2}$ to denote the assignment $\set{x_1 \mapsto \top, x_2 \mapsto \bot}$.
We use $\assignment \sqcup \assignment'$ as the update of partial assignment $\assignment$ with $\assignment'$, formally defined as \qquad
$
  (\assignment \sqcup \assignment')(v) =
  \begin{cases}
    \assignment'(v) & \text{if } v \in \dom(\assignment')\enspace, \\
    \assignment(v) & \text{otherwise} \enspace.
  \end{cases}
$

We write $\assignment \sqsubseteq \assignment'$ if $\assignment(v) = \assignment'(v)$ for every $v \in \dom(\assignment)$.
To restrict the domain of an assignment $\assignment$ to a set of variables $V$, we write $\assignment|_V$.
We denote \emph{the set of assignments} of a set of variables $V$ by $\assignments(V)$.
A \emph{Skolem function} $f_y \colon \assignments(\dep(y)) \rightarrow \bool$ maps an assignment of the dependencies of $y$ to an assignment of $y$.
The truth of a DQBF $\Phi$ with matrix $\varphi$ is equivalent to the existence of a Skolem function $f_y$ for every variable $y$ of the existentially quantified variables $\mathcal{Y}$, such that substituting all existentials $y$ in $\varphi$ by their Skolem function $f_y$ results in a valid formula.
We use $\Phi[\assignment]$ to denote the replacement of variables bound by $\assignment$ in $\Phi$ with the corresponding value.

\smallparagraph{Relation to QBF in prenex form.}
In QBF the dependencies of a variable are implicitly determined by the universal variables that occur before the quantifier in the quantifier prefix.
This gives rise to the notion that QBF have a \emph{linear} quantifier prefix, whereas DQBF allows for partially ordered quantifiers.


\section{Lifting Clausal Abstraction}  \label{sec:algorithm-detail}

In this section, we lift clausal abstraction to DQBF.
We begin with a high level explanation of the algorithm for QBF and a discussion of the invariants that hold for QBF but are no longer valid for DQBF.
For each of those we identify the underlying problem and show how we need to modify clausal abstraction.
In the following subsections we then explain those extensions in detail.
For the remainder of this section, we assume w.l.o.g. that we are given a DQBF $\Phi$ with matrix $\varphi$, that $\varphi$ does not contain clauses with information forks, and that every clause is universally reduced.
If a formula contains information forks initially, they can be removed as described in Section~\ref{sec:refinement}.

The clausal abstraction algorithm assigns existential and universal variables, where the order of assignments is determined by the quantifier prefix, until all clauses in the matrix are satisfied or there is a conflict, i.e., a set of clauses that cannot be satisfied simultaneously.
Those variable assignments are generated by propositional formulas, one for every quantifier, which we call \emph{abstractions}.
In case of a conflict, the reason for this conflict is excluded by refining the abstraction at an outer quantifier.

The assignment order is based on the quantifier prefix. 
Thus, for QBF it holds that an existential variable is only assigned if its dependencies are assigned.
In DQBF, Henkin quantifiers allow us to introduce incomparable dependency sets, and hence, in general, there is no linear order of assignments.
We thus weaken this invariant by requiring that for every existential variable $y$, all of its dependencies have to be assigned before assigning $y$.
We ensure this by creating a graph-based data structure, the dependency lattice, described in Section~\ref{sec:dependency-lattice}.
%
As an immediate consequence, and in contrast to QBF, an existential variable may be assigned different values depending on assignments to non-dependencies, and we call this phenomenon a \emph{spurious dependency}.
To eliminate those spurious dependencies, we enhance the certification approach of clausal abstraction~\cite{conf/fmcad/RabeT15} to build, incrementally, a constraint system that enforces that an existential variable only depends on its dependencies.
These \emph{consistency requirements} represent partial Skolem functions.
Section~\ref{sec:consistency} describes how the consistency requirements are derived, how they are integrated in the algorithm, and when they are invalidated.

We build an abstraction for every existential quantifier $\exists Y$, splitting every clause $C$ of the matrix into three parts, based on whether a literal $l \in C$ is (1) a dependency, (2) a literal of a variable in $Y$, or (3) neither of the two.
Section~\ref{sec:abstraction} gives a formal description of the abstraction.
As mentioned, all dependencies of $Y$ must be assigned before we query the abstraction of the quantifier $\exists Y$ for a candidate assignment of variables $Y$.
From the perspective of this abstraction, assignments to non-$Y$ variables are equivalent when they satisfy the same set of clauses.
Vice versa, the only information that matters for other abstractions is the set of clauses satisfied by variables $Y$ or their dependencies.
The abstraction for $Y$ therefore defines a set of interface variables consisting of \emph{satisfaction variables} and \emph{assumption variables}, one for every clause $C$, where the satisfaction variable indicates whether the clause is satisfied by a dependency of $Y$ and the assumption variable indicates whether $C$ must still be satisfied by variables outside of $Y$.
Conflicts are represented by a set of assumption variables that turned out to be not satisfiable only by variables outside of $Y$.
Refinements are clauses over those assumption variables, requiring that at least one of those contained clauses is satisfied by an assignment to $Y$.

Those refinements correspond to conflict clauses in search-based algorithms and can be formalized as derived clauses in the $\qres$ calculus~\cite{conf/cav/Tentrup17}.
Since $\qres$ is incomplete for DQBF and the incompleteness can be characterized by clauses with information fork, we check if a conflict clause derived by the algorithm contains such a fork.
If this is the case, we \emph{split} this clause into a set of clauses that are fork-free.
As a byproduct, new existential variables are created.
We show in Section~\ref{sec:refinement} how clauses with information fork are split and how the clausal abstraction algorithm is extended to \emph{incrementally} accept new clauses and variables.

\begin{example}
\label{ex:running}
We will use the following formula with the dependency sets $\set{x_1}$, $\set{x_2}$, and $\set{x_1,x_2}$ as a running example.
  \begin{align*}
    &\forall x_1, x_2 \ldot \exists y_1(x_1) \ldot \exists y_2(x_2) \ldot \exists y_3(x_1,x_2) \ldot \\
    &\underbrace{ (x_1 \lor \overlineindex{x_2} \lor \overlineindex{y_2} \lor y_3) }_{C_1}
      \underbrace{ (\overlineindex{x_1} \lor y_2 \lor y_3) }_{C_2}
      \underbrace{ (\overlineindex{y_1} \lor x_2 \lor \overlineindex{y_3}) }_{C_3}
      \underbrace{ (y_1 \lor \overlineindex{y_3}) }_{C_4}
      \underbrace{ (x_1 \lor y_1) }_{C_5}
  \end{align*}
\end{example}

\subsection{Dependency Lattice and Quantifier Levels} \label{sec:dependency-lattice}

To lift clausal abstraction to DQBF, we need to deal with \emph{partially} ordered dependency sets.
Given a DQBF $\Phi$, the algorithm starts with closing the dependency sets under intersection, which can also be described as building the meet-semilattice $\tuple{\mathcal{H}, \subseteq}$.
That is, $\mathcal{H}$ contains all dependency sets of variables in $\Phi$ and we add $H \cap H'$ to $\mathcal{H}$ for every $H, H' \in \mathcal{H}$ until a fixed point is reached.
We call this meet-semilattice the \emph{dependency lattice}. 
For our running example, we have to add the empty dependency set, resulting in the dependency lattice depicted on the left of Fig.~\ref{fig:quantifier-semi-lattice-example}.
In addition to the dependency sets $\mathcal{H}$ and the edge relation $\subseteq$, we depict the existential variables next to their dependency sets.

\begin{figure}[t]
  \begin{minipage}[c]{0.67\textwidth}
    \centering
    \begin{tikzpicture}[thick]
  \tikzstyle{node}=[draw, rectangle, rounded corners, inner sep=5pt,scale=0.82]
  \node[node] (root) {$\emptyset$};
  \node[node,below left=0.4 and 0.3 of root,label=left:$y_1$] (lhs) {$\set{x_1}$};
  \node[node,below right=0.4 and 0.3 of root,label=left:$y_2$] (rhs) {$\set{x_2}$};
  \node[node,below=2 of root,label=left:$y_3$] (full) {$\set{x_1,x_2}$};
  
  \draw (root) -- (lhs)
        (root) -- (rhs)
        (lhs) -- (full)
        (rhs) -- (full)
        ;
\end{tikzpicture}
    \quad
    \begin{tikzpicture}[thick]
  \tikzstyle{node}=[draw, rectangle, rounded corners, inner sep=5pt,scale=0.82]
  \node[node] (root) {0: $\tuple{\exists,\emptyset,\emptyset}$};
  \node[node,below=0.2 of root] (univ) {1: $\tuple{\forall, \set{x_1,x_2}}$};
  \node[node,below=0.2 of univ] (incomp) {2: $\tuple{\exists,\set{y_1}, \{x_1\}}$, $\tuple{\exists,\set{y_2}, \{x_2\}}$};
  \node[node,below=0.2 of incomp] (full) {3: $\tuple{\exists,\set{y_3}, \{x_1,x_2\}}$};
  
\end{tikzpicture}
  \end{minipage}\hfill
  \begin{minipage}[c]{0.3\textwidth}
    \caption{Dependency lattice (left) and quantifier levels (right) for the DQBF given in Example~\ref{ex:running}.} \label{fig:quantifier-semi-lattice-example}
  \end{minipage}
\end{figure}
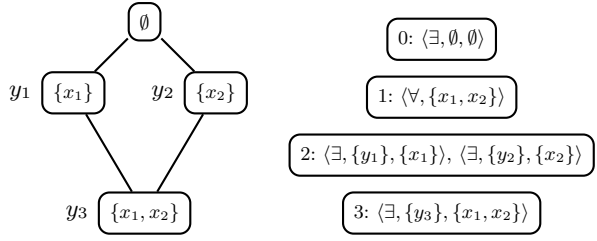

\smallparagraph{Quantifier Levels and Nodes.}
We continue with building the data structure on which the algorithm operates.
A \emph{node} binds a variable of the DQBF.
A universal node $\tuple{\forall, X}$ binds universal variables $X$ and an existential node $\tuple{\exists, Y, H}$ binds existential variables $Y$ with dependency set $H$.
Nodes are grouped together in \emph{quantifier levels}, where each universal level contains exactly one universal node and existential levels may contain multiple existential nodes.
We index levels by natural numbers $i$, starting with $0$.
On the right of Fig.~\ref{fig:quantifier-semi-lattice-example} is an example for the data structure obtained from the dependency lattice on its left.
Before describing the construction of quantifier levels, we state their invariants.
For some node $N$, let $\bound_\forall(N)$ be the set of variables bound at $N$, i.e., the union of all $X$ where $\tuple{\forall, X}$ is in a level with smaller index than node $N$.
Let $\bound_\exists(N)$ be the analogously defined set of bound existential variables.
The set of bound variables is $\bound(N) \coloneqq \bound_\exists(N) \mathbin{\dot\cup} \bound_\forall(N)$.

\begin{proposition}  \label{thm:quantifier-levels}
  The quantifier levels data structure has the following properties.
  \begin{enumerate}
    \item Every variable is bound excactly once, i.e., for every variable $v$ in $\Phi$, there is exactly one node $\tuple{\forall, X}$ or $\tuple{\exists, Y, H}$ such that $v \in X$ or $v \in Y$.\label{thm:bound-once}
    \item Every pair of nodes $\tuple{\exists, Y, H}$ and $\tuple{\exists, Y', H'}$ with $Y \neq Y'$ contained in an existential level have incomparable dependencies, i.e., $H \not\subseteq H'$ and $H \not\supseteq H'$.\label{thm:level-incomparable}
    \item For every pair of nodes $\tuple{\exists, Y_i, H_i}$ and $\tuple{\exists, Y_j, H_j}$ contained in existential levels $i$ and $j$ with $i < j$, it holds that either $H_i \subset H_j$ or $H_i$ and $H_j$ are incomparable.\label{thm:smaller-level}
    \item For every existential node $\tuple{\exists, Y, H}$ it holds that $H \subseteq \bound_\forall(\tuple{\exists, Y, H})$.\label{thm:over-approx}
    \item There is a unique maximal $\tuple{\exists, Y, H}$ with $H \supset H'$ for every other $\tuple{\exists, Y', H'}$.\label{thm:unique-maximal}
  \end{enumerate}
\end{proposition}
In the following, we describe the construction of quantifier levels from a dependency lattice.
Every element of the dependency lattice $H \in \mathcal{H}$ makes one existential node, $\tuple{\exists, Y, H}$, where $Y$ is the set of existential variables with dependency set $H$, i.e. $\dep(y) = H$ for all $y \in Y$.
Some existential nodes (like the root node in our example) may thus be initially empty.
The existential levels are obtained by an antichain decomposition of the dependency lattice (satisfying Proposition~\ref{thm:quantifier-levels}.\ref{thm:level-incomparable} and \ref{thm:quantifier-levels}.\ref{thm:smaller-level}).
If the dependency lattice does not contain a unique maximal element, we add an empty existential node $\tuple{\exists, \mathcal{X}, \emptyset}$ (Proposition~\ref{thm:quantifier-levels}.\ref{thm:unique-maximal}).

Universal variables are placed in the universal node just before the existential level they first appear in as a dependency.
This is achieved by a top-down pass through the existential quantifier levels, adding a universal level with node $N = \tuple{\forall, X}$ before existential level with nodes $\tuple{\exists,Y_1, H_1},\dots,\tuple{\exists,Y_k, H_k}$ such that $X = \left( \bigcup_{1 \leq i \leq k} H_i \right) \setminus \bound_\forall(N)$ (Proposition~\ref{thm:quantifier-levels}.\ref{thm:over-approx}).
Empty universal levels $\tuple{\forall, \emptyset}$ are omitted.
Level numbers follow the inverse order of the dependency sets, such that the ``outer'' quantifiers have smaller level numbers than the ``inner'' quantifiers; see Fig.~\ref{fig:quantifier-semi-lattice-example}.

If the formula is a QBF, it holds that $\bound_\forall(\tuple{\exists, Y, H}) = H$. 
For QBF, this construction yields a strict alternation between universal and existential levels, but for DQBF existential levels can succeed each other, as shown in Fig.~\ref{fig:quantifier-semi-lattice-example}.

\smallparagraph{Algorithmic Overview.}
The overall approach of the algorithm is to construct a propositional formula $\theta$ for every node, that represents which clauses it can satisfy (for existential nodes) or falsify (for universal nodes).
We describe their initialization in detail below.
In every iteration of the loop in algorithm \Call{Solve}{} (Fig.~\ref{alg:main}) the variable assignment $\assignment_V$ is extended (case $\mathsf{CandidateFound}$), which we assume to be globally accessible, or node abstractions are refined by adding an additional clauses (case $\mathsf{Conflict}$).

The nodes are responsible for determining candidate assignments to the variables bound at that node, or to give a reason why there is no such assignment.
If a node is able to provide a candidate assignment, we proceed to the successor level (Fig.~\ref{alg:main}, line~\ref{line:incLevel}).
A conflict occurs when the algorithm determines that the current assignment $\assignment_V$ definitely violates the formula (unsat conflict) or satisfies it (sat conflict).
When conflicts are inspected (explained in Section~\ref{sec:solving-nodes}), they indicate a level that tells the main loop how far we have to jump back (Fig.~\ref{alg:main}, line~\ref{line:jumpToEarlierLevel}).
The last alternative in the main loop is that we have found a result, which allows us to terminate (line~\ref{line:main-result}).

\begin{figure}[t]
  \fontsize{8}{10}\selectfont
  \begin{algorithmic}[1]
    \Procedure{Solve}{DQBF $\Phi$}
      \State $\mathit{levels} \gets$ build quantifier levels
      \State initialize every node in $\mathit{levels}$, i.e., build abstraction $\theta$, set $\mathit{entries} \gets []$
      \label{line:InitSAT}
      \State $\assignment_V \gets \set{}$, $\mathit{lvl} \gets 0$
      \Loop
        \Match {$\Call{SolveLevel}{\mathit{lvl}}$}
          \Case {$\mathsf{CandidateFound}$}  \label{fig:main-candidate-found}
            $\mathit{lvl} \gets \mathit{lvl} + 1$
            \label{line:incLevel}
          \EndCase
          \Case {$\mathsf{Conflict}(\mathit{jmpBackToLvl})$}  \label{fig:main-conflict}
            $\mathit{lvl} \gets \mathit{jmpBackToLvl}$
            \label{line:jumpToEarlierLevel}
          \EndCase
          \Case {$\mathsf{Result}(\mathit{res})$}
              \label{line:main-result}
            \Return $\mathit{res}$
          \EndCase
        \EndMatch
      \EndLoop
    \EndProcedure
  \end{algorithmic}
  \caption{Main solving algorithm that iterates over the quantifier levels.}  
  \label{alg:main}
\end{figure}

\subsection{Initialization of the abstractions $\theta$} \label{sec:abstraction}

The formula $\theta$ for each node represents how the node's variables interact with the assignments on other levels.
The algorithm guarantees that whenever we generate a candidate assignment for a node, all variables on outer (=smaller) levels have a fixed assignment, and thus some set of clauses is satisfied already.
Existential nodes then try to satisfy more clauses with their assignment, while universal nodes try to find an assignment that makes it harder to satisfy all clauses.
%
An existential variable $y$ may not only depend on assignments of its dependencies, but also on assignments of existential variables with strict smaller dependency as they are in a strictly smaller level (see Section~\ref{sec:dependency-lattice}) and thus are assigned before $y$.
We call this the extended dependency set, written $\exdep(y)$, and it is defined as $\dep(y) \mathbin{\dot\cup} \set{y' \in \mathcal{Y} \mid \dep(y') \subset \dep(y)}$.
For a set $Y \subseteq \mathcal{Y}$, we define $\exdep(Y) = \bigcup_{y \in Y} \exdep(y)$.

The interaction of abstractions is established by a common set of \emph{clause satisfaction} variables $S$, one variable $s_i \in S$ for every clause $C_i \in \matrix$.
Given some existential node $\tuple{\exists, Y, H}$ with extended dependency set $D = \exdep(Y)$ and assignment $\assignment_V$ of outer variables $V$ (w.r.t.~$\exists Y$, i.e., $V = \bound(\tuple{\exists, Y,H})$).
For every clause $C_i \in \matrix$ it holds that if $s_i$ is assigned to true, one of its dependencies has satisfied the clause, that is, $\assignment_V \models C_i|_D$.
Thus, an assignment of the satisfaction variables $\assignment_S$ is an abstraction of the concrete variable assignment $\assignment_V$ as multiple assignments could led to the same satisfied clauses.

For universal quantifiers, this abstraction is sufficient as the universal player tries to satisfy as few clauses as possible.
For existential quantifiers, however, the existential player can choose to either satisfy the clause directly or \emph{assume} that the clause will be satisfied by an inner quantifier.
Thus, we add an additional type of variables $A$, called \emph{assumption variables}, with the intended semantics that $a_i$ is set to false at some existential quantifier $\exists Y$ implies that the clause $C_i$ is satisfied at this quantifier (either by an assignment $\assignment_Y$ to variables $Y$ of the current node or an assignment of dependencies represented by an assignment $\assignment_S$ to the satisfaction variables $S$), formally, $\assignment_V \mathbin{\dot\sqcup} \assignment_Y \models C_i|_{D \mathbin{\dot\cup} Y}$ if $a_i$ is false.

We continue by defining the abstraction that implements this intuition.
Formally, for every node $\tuple{\exists, Y, H}$ and every clause $C_i$, we define
$C_i^< \coloneqq \set{l \in C_i \mid \var(l) \in \exdep(Y)}$ as the set of literals on which the current node may depend, 
$C_i^= \coloneqq \set{l \in C_i \mid \var(l) \in Y}$ as the the set of literals which the current node binds, and
$C_i^> \coloneqq \set{l \in C_i \mid \var(l) \not\in \exdep(Y) \cup Y}$ as the set of literals on which the current node may not depend.
By definition, it holds that $C = C_i^< \mathbin{\dot\cup} C_i^= \mathbin{\dot\cup} C_i^>$.
The clausal abstraction $\theta_Y$ for this node is defined as $\bigwedge_{C_i \in \varphi} \left( a_i \lor s_i \lor C_i^= \right)$.
Note, that $s_i$ and $a_i$ are omitted if $C_i^< = \emptyset$ and $C_i^> = \emptyset$, respectively.

Over time, the algorithm calls each node potentially many times for candidate assignments, and it adds new clauses learnt from refinements.
The new clauses for existential nodes will only contain literals from assumption variables $L \subseteq A$, representing sets of clauses that together cannot be satisfied by the inner levels.
The refinement $\bigvee_{a_i \in L} \overlineindex{a_i}$ ensures that some clause $C_i$ with $a_i \in L$ is satisfied at this node.

Universal nodes $\tuple{\forall, X}$ have the objective to falsify clause.
We define the abstraction $\theta_X$ for this node as $\bigwedge_{C_i \in \varphi} \left( s_i \lor \neg C_i^= \right) = \bigwedge_{C_i \in \varphi} \left( s_i \lor \bigwedge_{l \in C_i^=} \overline{l} \right)$.
Observe that universal nodes do not have separate sets of variables $A$ and $S$, but just one copy $S$.
This is just a minor simplification, exploiting the formula structure of universal nodes.
Note that $s_i$ set to false implies that $\assignment_X$ falsifies the literals in the clause, that is, $\assignment_X \models \neg C_i^=$.
Refinements are represented as clauses $\bigvee_{s_i \in L} \overlineindex{s_i}$ over literals in $S$.

In our running example, clauses 3--5 $(\overlineindex{y_1} \lor x_2 \lor \overlineindex{y_3})(y_1 \lor \overlineindex{y_3}) (x_1 \lor y_1)$ are represented at node $\tuple{\exists, \set{y_1}, \set{x_1}}$ by clauses $(a_3 \lor \overlineindex{y_1}) (a_4 \lor y_1) (y_1)$.
Note especially, that $x_2 \notin \exdep(y_1) = \set{x_1}$, thus there is no $s$ variable in the first encoded clause, despite $x_2$ being assigned earlier in the algorithm (Fig.~\ref{fig:quantifier-semi-lattice-example}).

\subsection{Solving Levels and Nodes} \label{sec:solving-nodes}

\textsc{SolveLevel} in Fig.~\ref{alg:qlevel} directly calls \textsc{Solve}$_\forall$ or \textsc{Solve}$_\exists$ on all the nodes in the level.
For existential levels, if any node returns a conflict, the level returns that conflict (Fig.~\ref{alg:qlevel}, line~\ref{line:returnconflict}).

\begin{figure}[t]
  \begin{minipage}[c]{0.67\textwidth}
    \fontsize{8}{10}\selectfont
    \begin{algorithmic}[1]
      \Procedure{SolveLevel}{$\mathit{lvl}$}
        \If{$\mathit{lvl}$ is universal}
          \Return $\Call{Solve$_\forall$}{\mathit{levels}[\mathit{lvl}]}$
        \EndIf
        \For {\textbf{each} node $n$ in $\mathit{levels}[\mathit{lvl}]$}
          \If{$\Call{Solve$_\exists$}{n} = \mathsf{Conflict}(\mathit{jmpBackToLvl})$}
            \State \Return $\mathsf{Conflict}(\mathit{jmpBackToLvl})$
            \label{line:returnconflict}
          \EndIf
        \EndFor 
        \State \Return $\mathsf{CandidateFound}$
      \EndProcedure
    \end{algorithmic}
  \end{minipage}\hfill
  \begin{minipage}[c]{0.3\textwidth}
    \caption{Algorithm for solving a quantifier level by iterating over the contained nodes.}  
    \label{alg:qlevel}
  \end{minipage}
\end{figure}

We assume a SAT solver interface $\Call{sat}{\psi, \assignment}$ for matrices $\psi$ and assumptions (represented by an assignment) $\assignment$.
It returns either $\mathsf{Sat}(\assignment')$, which means the formula is satisfiable with assignment $\assignment' \sqsupseteq \assignment$, or $\mathsf{Unsat}(\assignment')$, which means the formula is unsatisfiable and $\assignment' \sqsubseteq \assignment$ are the failed assumptions (i.e. the unsat core), that is, $\Call{sat}{\psi, \assignment'}$ is unsatisfiable as well.

We process universal and existential nodes with the two procedures shown in Fig.~\ref{alg:node-solving}.
The SAT solvers generate a candidate assignment to the variables (lines~\ref{line:callSATexistentialnode} and~\ref{line:callSATuniversalnode}) of that node, which is then used to extend the (global) assignment $\assignment_V$ (lines~\ref{line:update-assignments-existential} and \ref{alg:update-assignments-universal}).
In case the SAT solver returns $\mathsf{Unsat}$, the unsat core represents a set of clauses that cannot be satisfied (for existential nodes) or falsified (for universal nodes).
The unsat core is then used to refine an outer node (lines~\ref{line:refine1} and~\ref{line:refine2}) and we proceed with the level returned by $\Call{refine}{}$.

\smallparagraph{Solving Existential Nodes.}
There are some differences in the handling of existential and universal nodes that we look into now.
The linear ordering of the levels in our data structure means that there may be a variable assigned that an existential node must not depend on.
We therefore need to \emph{project} the assignment $\assignment_V$ to those variables in the node's dependency set.
We define a function $\prj_\exists \colon 2^\mathcal{Y} \times \assignments(V) \to \assignments(S)$ that maps variable assignments $\assignment_V$ to assignments of satisfaction variables $S$ such that $s_i$ is set to true if, and only if, some literal $l \in C_i^<$ is assigned positively by $\assignment_V$.
Thus, the projection function only considers actual dependencies of $\tuple{\exists, Y, H}$:

{\centering$
  \prj_\exists(Y, \assignment_V)(s_i) =
  \begin{cases}
    \top & \text{if } \assignment_V \models C_i^< \\
    \bot & \text{otherwise}
  \end{cases}
$\par}

For our running example, at node $\tuple{\exists, \set{y_2}, \set{x_2}}$, the projection for the first clause $C_1 = (x_1 \lor \overlineindex{x_2} \lor \overlineindex{y_2} \lor y_3)$ is
$\prj_\exists(\set{y_2}, \overlineindex{x_1} \bm{x_2})(s_1) = \prj_\exists(\set{y_2}, x_1 \bm{x_2})(s_1) = \bot$ and
$\prj_\exists(\set{y_2}, \overlineindex{x_1} \bm{\overlineindex{x_2}})(s_1) = \prj_\exists(\set{y_2}, x_1 \bm{\overlineindex{x_2}})(s_1) = \top$ because $C_1^< = (\overlineindex{x_2})$.

If the SAT solver returns a candidate assignment at the maximal existential node (i.e., the node on innermost level), we know that all clauses have been satisfied, and we have therefore refuted the candidate assignment of some universal node.
This is handled by calling \textsc{Refine} in line~\ref{line:refineleaf}.
For existential nodes we additionally have to check for consistency, which we discuss in Section~\ref{sec:consistency} (called in line~\ref{line:check-skolem}).  

\algrenewcommand\algorithmicindent{1.0em}%
\begin{figure}[t]
  \begin{multicols}{2} \fontsize{8}{10}\selectfont
  \begin{algorithmic}[1]
    \Procedure{Solve$_\exists$}{$\mathit{node} \equiv \tuple{\exists, Y, H}$}
      \State $\assignment_{Y} \gets \Call{checkConsistency}{\assignment_V}$  \label{line:check-skolem}
      \State $\assignment_S \gets \prj_\exists(Y, \assignment_V)$
      \Match {$\Call{sat}{\theta_Y, \assignment_{Y} \sqcup \assignment_S}$} \label{line:callSATexistentialnode} 
        \Case{$\mathsf{Unsat}(\assignment)$}
          \State\Return $\Call{refine}{\unsat, \assignment|_S, \mathit{node}}$\label{line:refine1}\hspace{-40pt}\mbox{}
        \EndCase
        \Case{$\mathsf{Sat}(\assignment)$}
          update $\assignment_V$ with $\assignment|_Y$  \label{line:update-assignments-existential}
          \If{node is maximal element}
            \State \Return $\Call{refine}{\sat, \assignment_S, \mathit{node}}$\label{line:refineleaf}
          \EndIf
          \State \Return $\mathsf{CandidateFound}$
        \EndCase
      \EndMatch
    \EndProcedure
    \columnbreak
    \Procedure{Solve$_\forall$}{$\mathit{node} \equiv \tuple{\forall, X}$}
      \State $\assignment_S \gets \prj_\forall(X, \assignment_V)$
      \label{line:uproj}
      \Match {$\Call{sat}{\theta_X, \assignment_S}$}  \label{line:callSATuniversalnode}
        \Case{$\mathsf{Unsat}(\assignment)$}
          \State\Return $\Call{refine}{\sat, \assignment|_S, \mathit{node}}$
          \label{line:refine2}
        \EndCase
        \Case{$\mathsf{Sat}(\assignment)$}
        \label{line:nodeCaseSAT}
          update $\assignment_V$ with $\assignment|_X$  \label{alg:update-assignments-universal}
          \State \Return $\mathsf{CandidateFound}$  \label{alg:next-nodes}
        \EndCase
      \EndMatch
    \EndProcedure
  \end{algorithmic}
  \end{multicols}
  \vspace{-10pt}
  \caption{Process existential and universal nodes.}
  \label{alg:node-solving}
\end{figure}
\algrenewcommand\algorithmicindent{1.5em}%

\smallparagraph{Solving Universal Nodes.}
Similar to the projection for existential nodes, we need an (almost symmetric) projection for universal nodes (line~\ref{line:uproj}).
It has to differ slightly from $\prj_\exists$, because we use just one set of variables $S$ for universal nodes.
A universal quantifier cannot falsify the clause if it is already satisfied. 

{\centering$
  \prj_\forall(X, \assignment_V)(s_i) =
  \begin{cases}
    \top & \text{if } \assignment_V \models C_i|_{\bound{\tuple{\forall, X}}} \\
    \mathit{undef} & \text{otherwise}
  \end{cases}
$\par}

\subsection{Refinement}  \label{sec:refinement}
Algorithm \Call{refine}{} in Fig.~\ref{alg:refinement} is called whenever there is a conflict, i.e. whenever it is clear that $\assignment_V$ satisfies the formula ($\sat$ conflict) or violates it ($\unsat$ conflict).
In case there is an unsat conflict at an existential node, we build the (universally reduced) conflict clause from $\assignment_S$~\cite{conf/cav/Tentrup17} in line~\ref{line:build-conflict-clause}.
If the clause is fork-free, we can apply the standard refinement for clausal abstraction~\cite{conf/fmcad/RabeT15} with the exception that we need to find the unique refinement node first (line~\ref{alg:determine-refinement-node}).
This backward search over the quantifier levels is shown in Fig.~\ref{alg:backward-search}.
For an $\unsat$ conflict, we traverse the levels backwards until we find an existential node that binds a variable contained in the conflict clause.
Because the conflict clause is fork-free, the target node of the traversal is unique.
For a $\sat$ conflict, we do the same for universal nodes but the uniqueness comes from the fact that universal levels are singletons.
We then add the refinement clause to the SAT solver at the corresponding node (lines~\ref{line:refine-unsat} and \ref{line:refine-sat}) and proceed.
For $\sat$ conflicts, we have to additionally learn consistency requirements at existential nodes (line~\ref{alg:learn-skolem}) that make sure that the node produces the same result if the assignment (restricted to the dependencies of that node) repeats.
In case the conflict propagated beyond the root node, we terminate with the given result.

\begin{figure}[t]
  \fontsize{8}{10}\selectfont
  \begin{algorithmic}[1]
    \Procedure{Refine}{$\mathit{res}, \assignment_S, \mathit{node}$}
      \If{$\mathit{res} = \unsat$}
        \State $C_\mathit{conflict} = \bigcup_{C_i \in \varphi, \assignment_S(s_i) = \bot} C_i|_{\bound(\mathit{node})}$  \Comment{$C_\mathit{conflict}$ is universally reduced} \label{line:build-conflict-clause}
        \If{$C_\mathit{conflict}$ contains information fork}
          \State fork elimination $\Rightarrow$ add clauses and variables, update abstractions $\theta$
          \State \Call{resetConsistency}{} for all nodes   \label{line:refinement-reset-skolem}
          \State \Return $\mathsf{Conflict}(\mathit{lvl} = 0)$
        \EndIf
      \EndIf
      \If{$\mathit{next} \gets \Call{DetermineRefinementNode}{\mathit{res}, \assignment_S, \mathit{node}.\mathit{level}}$}  \label{alg:determine-refinement-node}
        \State \Return $\mathsf{Conflict}(\mathit{next}.\mathit{level})$
      \Else
        \Comment{conflict at outermost $\exists/\forall$ node}
        \State \Return $\mathsf{Result}(\mathit{res})$
      \EndIf  
    \EndProcedure
  \end{algorithmic}
  \caption{Refinement algorithm that applies Fork Extension in case of information forks.}
  \label{alg:refinement}
\end{figure}

\smallparagraph{Fork Extension.}
In case that the conflict clause contains a fork, we apply Fork Extension~\cite{conf/sat/Rabe17}\footnote{Fork Extension as introduced in~\cite{conf/sat/Rabe17} is incomplete for general DQBF. However, it is complete for a normal form of DQBF. We refer to Appendix~\ref{sec:proofsystem} for details.}.
\emph{Fork Extension} allows us to split a clause $C_1\vee C_2$ by introducing a fresh variable $y$. 
The dependency set of $y$ is defined as the intersection $\dep(C_1) \cap \dep(C_2)$ and represents that the question whether $C_1$ or $C_2$ satisfies the original clause needs to be resolved based on the information that is available to both of them. 
Fork Extension is usually only applied when $C_1$ and $C_2$ have incomparable dependencies ($\dep(C_1) \nsubseteq \dep(C_2)$ and $\dep(C_1) \nsupseteq \dep(C_2)$), as only then the dependency set of $y$ is smaller than those of $C_1$ and of $C_2$. 
The formal definition of the rule is 
\begin{equation*}
  \infer[\forkext]
  {
    \exists y (\dep(C_1) \cap \dep(C_2)) \ldot~ C_1 \!\cup\! \set{y} ~\wedge~ C_2 \!\cup\! \set{\overline{y}}
  }
  {
    C_1 \cup C_2
    \qquad
    y \text{ is fresh}
  }
\end{equation*}

\begin{example}  \label{ex:fork-extension}
  As an example of applying Fork Extension, consider the quantifier prefix
  $\forall x_1 x_2 \ldot \exists y_1(x_1) \ldot \exists y_2(x_2)$ and clause $(\overlineindex{x_1} \lor y_1 \lor y_2)$.
  Applying $\forkext$ with the decomposition $C_1 = \set{\overlineindex{x_1}, y_1}$ and $C_2 = \set{y_2}$ results in the clauses $(\overlineindex{x_1} \lor y_1 \lor \bm{y_3})(\bm{\overlineindex{y_3}} \lor y_2)$ where $y_3$ is a fresh existential variable with dependency set $\dep(y_3) = \emptyset$ ($\dep(C_1) = \set{x_1}$ and $\dep(C_2) = \set{x_2}$).
\end{example}
After applying Fork Extension, we encode the newly created clauses and variables within their respective nodes.
We update the abstractions with those fresh variables and clauses as for the initial abstraction discussed in Section~\ref{sec:abstraction}.
Additionally, we reset learned Skolem functions as they may be invalidated by the refinement (Fig.~\ref{alg:refinement}, line~\ref{line:refinement-reset-skolem}).

\begin{figure}[t]
  \fontsize{8}{10}\selectfont
  \begin{algorithmic}[1]
    \Procedure{DetermineRefinementNode}{$\mathit{res}, \assignment_S, \mathit{lvl}$}
      \While {$\mathit{lvl} \geq 0$}
        \If {$\mathit{res} = \unsat$ and $\mathit{lvl}$ is existential}
          \For {node $\tuple{\exists, Y, H}$ in $\mathit{levels}[\mathit{lvl}]$} \Comment{check if $Y$ is contained in conflict clause}
            \If {$C_i|_Y \neq \emptyset$ for some $C_i \in \matrix$ with $\assignment_S(s_i) = \bot$}
              \State $\theta_Y \gets \theta_Y \land \bigvee_{C_i \in \varphi, \assignment_S(s_i) = \bot} \overlineindex{a_i}$  \Comment{refine abstraction} \label{line:refine-unsat}
              \State \Return $\tuple{\exists, Y, H}$
            \EndIf
          \EndFor
        \ElsIf {$\mathit{res} = \sat$ and $\mathit{lvl}$ is universal with node $\tuple{\forall, X}$} 
          \If {$C_i|_X \neq \emptyset$ for some $C_i \in \matrix$ with $\assignment_S(s_i) = \top$}
            \State $\theta_X \gets \theta_X \land \bigvee_{C_i \in \varphi, \assignment_S(s_i) = \top} \overlineindex{s_i}$  \Comment{refine abstraction} \label{line:refine-sat}
            \State \Return $\tuple{\forall, X}$
          \EndIf
        \ElsIf {$\mathit{res} = \sat$ and $\mathit{lvl}$ is existential}  \Comment{add consistency requirements}
          \State \Call{learnEntry}{$N$} \textbf{ for each } $N = \tuple{\exists, Y, H}$ in $\mathit{levels}[\mathit{lvl}]$ with $H \subset \bound_\forall(N)$ \label{alg:learn-skolem}
        \EndIf
        \State $\mathit{lvl} \gets \mathit{lvl} - 1$
      \EndWhile
      \State \Return $\mathit{res}$
    \EndProcedure
  \end{algorithmic}
  \caption{Backward search algorithm to determine refinement node.}
  \label{alg:backward-search}
\end{figure}

\subsection{Consistency Requirements}  \label{sec:consistency}
%
The algorithm described so far produces correct refutations in case the DQBF is false.
For positive results, the \emph{consistency} of Skolem functions of \emph{incomparable} existential variables may be violated.
Consider for example the formula
$\forall x_1 \forall x_2 \ldot \exists y_1(x_1) \ldot \exists y_2(x_2) \ldot \exists y_3(x_1,x_2) \ldot \varphi$
and assume that for the assignment $\overlineindex{x_1} \overlineindex{x_2}$, there is a corresponding satisfying assignment $\overlineindex{y_1} \overlineindex{y_2} \overlineindex{y_3}$.
If the next assignment is $\overlineindex{x_1} \bm{x_2}$, then the assignment to $y_1$ has to be the same as before ($y_1 \to \bot$) as the value of its sole dependency $x_1$ is unchanged.

We enhance the certification capabilities of clausal abstraction~\cite{conf/fmcad/RabeT15} to build consistency requirements that represent partial Skolem functions in our algorithm during solving.
We incrementally build a list of \emph{entries}, where the first component in an entry is a propositional formula over the dependencies and the second component is the corresponding assignment $\assignment_Y$.
Before generating a candidate assignment in \Call{Solve$_\exists$}{}, we call \Call{checkConsistency}{} (Fig.~\ref{alg:skolem-learning}) to check if the assignment $\assignment_Y$ for the given assignment $\assignment_V$ of dependencies is already determined, by iterating through the learned entries (Fig.~\ref{alg:skolem-learning}, lines~\ref{line:skolem-learning-iter-start}--\ref{line:skolem-learning-iter-end}).
If it is the case, we get an assignment $\assignment_Y$ that is then assumed for the candidate generation.
Note that in this case, the $\Call{sat}{}$ call in line~\ref{line:callSATexistentialnode} of \Call{Solve$_\exists$}{} is guaranteed to return $\mathsf{Sat}$ (we already verified this assignment, otherwise it would not have been learned).
Further, consistency requirements are only needed for existential nodes $\tuple{\exists, Y, H}$ with $H \subset \bound_\forall(\tuple{\exists, Y, H})$, i.e., that observe an over-approximation of their dependency set.
For those nodes, the consistency requirements enforce that whenever two assignments of the dependencies are equal, the assignment of $\assignment_Y$ returns the same value as well.
We call \Call{resetConsistency}{} (Fig.~\ref{alg:skolem-learning}) to reset the consistency requirements in case we applied Fork Extension (Fig.~\ref{alg:refinement}, line~\ref{line:refinement-reset-skolem}) as the new clauses may affect already learned parts of the function.
We, further, have to reset the clauses learned at universal nodes (Fig.~\ref{alg:skolem-learning}, line~\ref{alg:reset-learned-universal}).

We \emph{learn} a new consistency requirement by calling \Call{learnEntry}{} (Fig.~\ref{alg:skolem-learning}) on the backward search on sat conflicts, that is in line~\ref{alg:learn-skolem} in Fig.~\ref{alg:backward-search}.
When we determine the refinement node for sat conflicts, we call $\Call{learnEntry}{}$ in every existential node $\tuple{\exists, Y, H}$ with $H \subset \bound_\forall(\tuple{\exists, Y, H})$ on the path to that node.
In our example, when the base case of $\tuple{\exists, \set{y_3}, \set{x_1,x_2}}$ returns (all clauses are satisfied, line~\ref{line:refineleaf} in Fig.~\ref{alg:node-solving}), we add consistency requirmenets at nodes $\tuple{\exists, \set{y_2}, \set{x_2}}$ and $\tuple{\exists, \set{y_1}, \set{x_1}}$ before refining at $\tuple{\forall, \set{x_1,x_2}}$.

\begin{figure}[t]
  \begin{minipage}[c]{0.67\textwidth}
    \fontsize{8}{10}\selectfont
    \begin{algorithmic}[1]    
      \Procedure{checkConsistency}{$\assignment_V$}
        \For{$(\mathit{cond}, \assignment_Y)$ in $\mathit{entries}$}\label{line:skolem-learning-iter-start}
          \If{$\Call{sat}{\mathit{cond}, \assignment_V}$ is $\mathsf{Sat}$}
            \Return $\assignment_Y$  \label{line:skolem-learning-iter-end}
          \EndIf
        \EndFor
        \State \Return empty assignment
      \EndProcedure
      \Procedure{resetConsistency}{}
        \State $\mathit{entries} \gets []$
        \State reset learned clauses at universal nodes  \label{alg:reset-learned-universal}
      \EndProcedure
      \Procedure{learnEntry}{$\mathit{node} \equiv \tuple{\exists, Y, H}$}
        \State let $\assignment_S$ and $\assignment_Y$ be from line~\ref{line:update-assignments-existential} of Fig.~\ref{alg:node-solving}.
        \State $\mathit{entries}.\mathit{push}((\bigwedge_{C_i \mid \assignment_S(s_i) = \top} C_i^<, \assignment_Y))$
      \EndProcedure 
    \end{algorithmic}
  \end{minipage}\hfill
  \begin{minipage}[c]{0.3\textwidth}
    \caption{Algorithms for handling consistency requirements.}  
    \label{alg:skolem-learning}
  \end{minipage}
\end{figure}

\subsection{Example}
We consider a possible execution of the presented algorithm on our running example.
For the sake of readability, we combine unimportant steps and focus on the interesting cases.
Assume the following initial assignment $\assignment_1 = x_1 \overlineindex{x_2} \overlineindex{y_1} \overlineindex{y_2}$ before node $N_\mathit{max} \equiv \tuple{\exists, \set{y_3}, \set{x_1,x_2}}$.
The result of projecting function $\prj_\exists(\set{y_3}, \assignment_1)$ is $s_1 \overlineindex{s_2} s_3 \overlineindex{s_4} s_5$ and the SAT solver (Fig.~\ref{alg:node-solving}, line~\ref{line:callSATexistentialnode}) returns $\mathsf{Unsat}(\assignment_1')$ with core $\assignment_1' = \overlineindex{s_2} \overlineindex{s_4}$ as there is no way to satisfy both clauses $(s_2 \lor y_3)$ and $(s_4 \lor \overlineindex{y_3})$ of the abstraction.
The refinement algorithm (Fig.~\ref{alg:refinement}) builds the conflict clause $C_\mathit{conflict} = C_2|_{\bound(N_3)} \cup C_4|_{\bound(N_3)} = (\overlineindex{x_1} \lor y_2 \lor y_1)$ at line~\ref{line:build-conflict-clause} which contains an information fork between $y_1$ and $y_2$.
We have already seen in Example~\ref{ex:fork-extension} that the fork can be eliminated resulting in fresh variable $y_4$ with $\dep(y_4) = \emptyset$ and the clauses 6 and 7 $(\overlineindex{x_1} \lor y_1 \lor y_4)(\overlineindex{y_4} \lor y_2)$.

Now, the root node contains variable $y_4$, for which we assume assignment $\set{y_4 \mapsto \top}$.
For the same universal assignment as before ($x_1 \overlineindex{x_2}$), the assignment of $y_2$ has to change to $\set{y_2 \mapsto \top}$ due to the newly added clause 7, leading to $\assignment_2 = x_1 \overlineindex{x_2} \overlineindex{y_1} y_2 y_4$ before node $N_\mathit{max}$.
The only unsatisfied clause is $C_4$ which can be satisfied using $\set{y_3 \mapsto \bot}$, leading to the base case (Fig.~\ref{alg:node-solving}, line~\ref{line:refineleaf}).
During refinement, we learn Skolem function entries $(x_1 \land y_4, \overlineindex{y_1})$ and $(\overlineindex{x_2} \land y_4, y_2)$ at nodes $\tuple{\exists, \set{y_1}, \set{x_1}}$ and $\tuple{\exists, \set{y_2}, \set{x_2}}$ as $\prj_\exists(\set{y_1},  x_1 \overlineindex{x_2})$ and $\prj_\exists(\set{y_2}, x_1 \overlineindex{x_2})$ assign $s_1$, $s_5$, $s_6$ and $s_1$, $s_6$ positively, respectively.

For the following universal assignment $\overlineindex{x_1} \overlineindex{x_2}$, the value of $y_2$ is already determined by the consistency requirements (Fig.~\ref{alg:node-solving}, line~\ref{line:check-skolem}) to be positive.
There is a continuation of the algorithm without further unsat conflict, determining that the instance is true.

\subsection{Correctness}  \label{sec:correctness}
%
We sketch the correctness argument for the algorithm, which relies on formal arguments regarding correctness and certification of the clausal abstraction algorithm~\cite{conf/fmcad/RabeT15} and the subsequent analysis of the underlying proof system~\cite{conf/cav/Tentrup17}.\footnote{A formal correctness proof is given in Appendix~\ref{sec:correctness-proof}.}
%
For soundness, the algorithm has to guarantee that existential variables are assigned consistently, that is for an existential variable $y$ with dependency $\dep(y)$ it holds that $f_y(\assignment) = f_y(\assignment')$ if $\assignment|_{\dep(y)} = \assignment'|_{\dep(y)}$ for every $\assignment$ and $\assignment'$.
Our algorithm maintains this property at every point during the execution by a combination of over-approximation and consistency requirements.
Completeness relies on the fact that the underlying proof system is refutationally complete for DQBF.
Progress is guaranteed as there are only finitely many different conflict clauses and, thus, only finitely many Skolem function resets.

\section{Evaluation}  \label{sec:evaluation}

We compare our prototype implementation, called $\dcaqe$\footnote{Available at \url{https://github.com/ltentrup/caqe}.}, against the publicly available DQBF solvers, $\idq$~\cite{conf/sat/FrohlichKBV14}, $\hqs$~\cite{conf/date/GitinaWRSSB15}, and $\iprover$~\cite{conf/cade/Korovin08}.
We ran the experiments on machines with a 3.6\,GHz quad-core Xeon processor with timeout and memout set to 10 minutes and 8\,GB, respectively.
We used the DQBF preprocessor $\hqspre$~\cite{conf/tacas/WimmerRM017} for every solver except $\hqs$.
We evaluate our solver on the DQBF case studies regarding reactive synthesis~\cite{conf/tacas/FaymonvilleFRT17} and the partial equivalence checking problem (PEC)~\cite{conf/date/GitinaWRSSB15,conf/sat/FinkbeinerT14}.

\begin{table}[t]
  \caption{Number of instances solved within 10\,min. For every solver, we give the number of solved instances overall ($\#$) and broken down by satisfiable ($\top$), unsatisfiable ($\bot$), and uniquely solved instances ($*$).}
  \label{tbl:experiment-results}
  \centering
  \scalebox{.95}{
  \begin{tabular}{lr||rrrr|rrrr|rrrr|rrrr}
    Benchmark & \# & \multicolumn{4}{c|}{$\dcaqe$} & \multicolumn{4}{c|}{$\idq$} & \multicolumn{4}{c|}{$\hqs$} & \multicolumn{4}{c}{$\iprover$} \\
    && $\#$ & $\top$ & $\bot$ & $*$ & $\#$ & $\top$ & $\bot$ & $*$ & $\#$ & $\top$ & $\bot$ & $*$ & $\#$ & $\top$ & $\bot$ & $*$ \\
    \hline\hline
    PEC1~\cite{conf/sat/FinkbeinerT14} & 1000  
    & \textbf{839} & 7 & \textbf{832} & \textbf{224}
    & 37 & 0 & 37 & 0
    & 636 & \textbf{10} & 626 & 32
    & 71 & 0 & 71 & 0 \\
    PEC2~\cite{conf/date/GitinaWRSSB15} & 720 
    & 342 & 71 & 271 & 12
    & 214 & 45 & 169 & 0
    & \textbf{401} & \textbf{104} & \textbf{297} & \textbf{60}
    & 288 & 60 & 228 & 0 \\
    BoSy~\cite{conf/tacas/FaymonvilleFRT17} & 1216
    & \textbf{1006} & \textbf{389} & \textbf{617} & \textbf{66}
    & 924 & 335 & 589 & 2
    & 735 & 231 & 504 & 0
    & 946 & 370 & 576 & 20 \\
    \hline
              & 2936 & \textbf{2187} & & & & 1175 & & & & 1772 & & & & 1305 
  \end{tabular}
  }
\end{table}

The first two benchmark sets consider the partial equivalence checking~(PEC) problem~\cite{conf/iccd/GitinaRSWSB13}, that is, the problem whether a circuit containing not-implemented (combinatorial) parts, so called ``black boxes'', can be completed such that it is equivalent to a reference circuit.
The inputs to the circuit are modeled as universally quantified variables and the outputs of the black boxes as existentially quantified variables.
Since the output of a black box should only depend on the inputs that are actually visible to the
black box, we need to restrict the dependencies of the existentially quantified variables to subsets of the universally quantified variables.
The benchmark sets PEC1 and PEC2 refers to~\cite{conf/sat/FinkbeinerT14} and \cite{conf/date/GitinaWRSSB15}, respectively.
The second case study (BoSy) considers the problem of synthesizing sequential circuits from specifications given in linear-time temporal logic~(LTL)~\cite{conf/tacas/FaymonvilleFRT17}.
The benchmarks were created using the tool BoSy~\cite{conf/cav/FaymonvilleFT17} and the LTL benchmarks from the Reactive Synthesis Competition~\cite{journals/corr/JacobsBBK0KKLNP16,journals/corr/abs-1711-11439}.
Each formula encodes the existence of a sequential circuit that satisfies the LTL specification.

The results are presented in Table~\ref{tbl:experiment-results}.
The PEC instances contain over-proport\-ionally many unsatisfiable instances and we conjecture that the differences between $\dcaqe$ and $\idq$/$\iprover$ can be explained by the effectiveness of the resolution-based refutations that $\dcaqe$ is based on.
$\hqs$ performs well on those benchmarks as well, which could be due to the fact that it implements the fast refutation technique~\cite{conf/sat/FinkbeinerT14} that was introduced alongside the benchmark set PEC1.
The reactive synthesis benchmark set is were $\dcaqe$ excels.
The benchmark set contains many easily solvable benchmarks, indicated by the high number of instances that are commonly solved by all solvers.
However, there are also a fair amount of hard instances and $\dcaqe$ solves significantly more of those than any other solver.
Further, we can see the effect mentioned in the introduction of infeasibility of expansion-based methods as shown by the result of $\hqs$.
The cactus plot given in Fig.~\ref{fig:cacuts-reactive-synthesis} shows that $\dcaqe$ makes more progress, especially with a larger runtime where the other solvers solve very few instances after 100s.
These results give rise to the hope that the scalability of more expressive synthesis approaches~\cite{conf/cav/FinkbeinerHLST18,conf/cav/CoenenFST19,journals/corr/FinkbeinerT15} can be improved by employing DQBF solving.

\begin{figure}[t]
  \centering
  \begin{tikzpicture}
    \begin{axis}[
      xlabel=\# solved instances,
      ylabel=time (sec.),
      width=.95\columnwidth,
      height=5cm,
      ymin=0,
      ymax=600,
      xmin=550,
      xmax=1010,
      legend entries={$\dcaqe$, $\iprover$, $\idq$, $\hqs$},
      mark size=0, 
        legend style={
          at={(.05,0.95)},
          anchor=north west,
          draw=none}]
        ]]
        
      \addplot+[blue,solid,mark=*,mark options={fill=blue},line width=1pt] table {plots/dcaqe-paper-cactus-dcaqe-hqspre-expansion-g0.dat};
      
      \addplot+[red,solid,mark=square*,mark options={fill=red},line width=1pt] table {plots/dcaqe-paper-cactus-iprover-hqspre-g0.dat};
      
      \addplot+[green,solid,mark=triangle*,mark options={fill=green},line width=1pt] table {plots/dcaqe-paper-cactus-idq-hqspre-g0.dat};
      
      \addplot+[orange,solid,mark=diamond*,mark options={fill=orange},line width=1pt] table {plots/dcaqe-paper-cactus-hqs-new-g0.dat};

    \end{axis}
  \end{tikzpicture}
  \vspace{-5pt}
  \caption{Cactus plot for the BoSy benchmark.}
  \label{fig:cacuts-reactive-synthesis}
  \vspace{-10pt}
\end{figure}
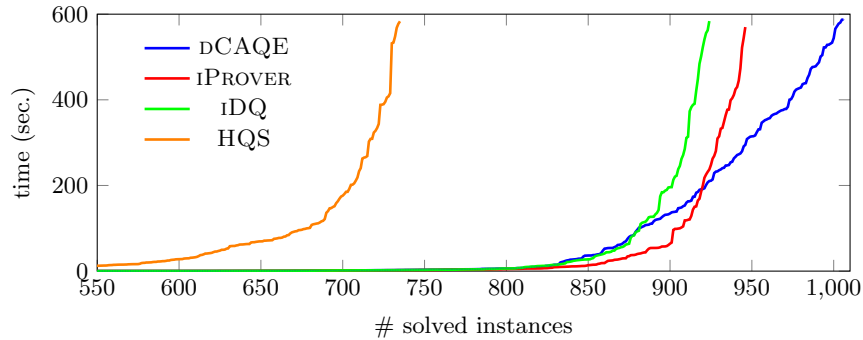

\section{Related Work}  \label{sec:related-work}

The satisfiability problem for DQBF was shown to be $\nexptime$-complete~\cite{dqbf-nexptime}.
Fr\"ohlich et al.~\cite{conf/pos/FrohlichKB12} proposed a first detailed solving algorithm for DQBF based on DPLL.
They already encountered many challenges of lifting QBF algorithms to DQBF, like Skolem function consistency, replay of Skolem functions, forks in conflict clauses, but solved them differently.
Their algorithm, called DQDPLL, has some similarities to our algorithm (in the same way that clausal abstraction and QDPLL share the same underlying proof system~\cite{conf/cav/Tentrup17}), but performs significantly worse~\cite{conf/pos/FrohlichKB12}.
We highlight a few differences which we believe to be crucial:
  (1) Our algorithm tries to maintain as much order as possible. Placing universal nodes at the latest possible allows us to apply the cheaper QBF refinement method more often.
  (2) We learn consistency requirements only if they have been verified to satisfy the formula, while DQDPLL learns them on decisions.
      Consequently, in DQDPLL, learned Skolem functions become part of the clauses, thus, making conflict analysis more complicated and less effective as they may be undone during solving.
      We keep the consistency requirements distinct from the clauses, all learned clauses at existential nodes are thus valid during solving.
  (3) Skolem functions in DQDPLL are represented as clauses representing truth-table entries, thus, become quickly infeasible.
      In contrast, we use a separate certification mechanism as in QBF solvers~\cite{conf/fmcad/RabeT15}.
iDQ~\cite{conf/sat/FrohlichKBV14} uses an instantiation-based algorithm which is based on the Inst-Gen calculus, a state-of-the art decision procedure for the effectively propositional fragment of first-order logic~(EPR), which is also $\nexptime$-complete.
HQS~\cite{conf/date/GitinaWRSSB15} is an expansion based solver that expands universal variables until the resulting instance has a linear prefix and applies QBF solving afterwards.
Bounded unsatisfiability~\cite{conf/sat/FinkbeinerT14} asserts the existence of a partial (bounded) expansion tree that guarantees that no Skolem function exists.
QBF preprocessing techniques have been lifted to DQBF~\cite{conf/sat/WimmerGNSB15,conf/tacas/WimmerRM017}.
Our solving technique is based on clausal abstraction~\cite{conf/fmcad/RabeT15} (also called clause selection~\cite{conf/ijcai/JanotaM15}) for QBF, which can provide certificates~\cite{conf/fmcad/RabeT15}.
Later, it was shown that refutation in clausal abstraction can be simulated by $\qres$~\cite{conf/cav/Tentrup17}.

\section{Conclusions}
\label{sec:conclusion}

We lifted the clausal abstraction algorithm to DQBF.
This algorithm is the first to exploit the new Fork Resolution proof system and it significantly increases performance of DQBF solving on synthesis benchmarks.
In particular, in the light of the past attempts to define search algorithms~\cite{conf/pos/FrohlichKB12} (which are closely related to clausal abstraction) for DQBF this is a surprising success.
It appears that the Fork Extension proof rule was the missing piece in the puzzle to build search/abstraction algorithms for DQBF.
\\[\smallskipamount]\textbf{Acknowledgments.}\quad
We thank Bernd Finkbeiner for his valuable feedback on earlier versions of this paper.
This work was partially supported by the German Research Foundation (DFG) as part of the Collaborative Research Center ``Foundations of Perspicuous Software Systems’' (TRR 248, 389792660) and by the European Research Council (ERC) Grant OSARES (No. 683300).

\bibliographystyle{splncs04}
\bibliography{main}

\newpage
\appendix
\section{Underlying Proof System}  \label{sec:proofsystem}

\begin{remark}
  We want to emphasize that the following does not impact the algorithm or the results presented in this paper.
  For the interested reader, we provide an extension of Fork Extension that is complete for general DQBF and show that Fork Extension is complete for a normal form of DQBF, which we dub multi-linear DQBF. This normal form covers all existing benchmarks.
\end{remark}

In this section we recall the Fork Resolution proof system, which underlies the algorithm proposed in this paper. 
We also discuss a problem with the completeness of Fork Resolution and suggest two ways to overcome the problem.
Fork Resolution consists of the well-known proof rules \emph{Resolution} and \emph{Universal Reduction} and introduces a new proof rule called \emph{Fork Extension}~\cite{conf/sat/Rabe17}. 

\emph{Resolution} allows us to merge two clauses as follows:
Given two clauses $C_1 \lor v$ and $C_2 \lor \neg v$, we call $(C_1 \lor v) \otimes_v (C_2 \lor \neg v) = C_1 \lor C_2$ their \emph{resolvent} with pivot $v$. 
The resolution rule states that $C_1\vee v$ and $C_2 \vee \neg v$ imply their resolvent.
\emph{Universal reduction} allows us to drop universal variables from clauses when none of the existential variables in that clause may depend on them. 
Let $C$ be a clause, let $l\in C$ be a literal of a universal variable, and let $\overline l\notin C$. 
If for all existential variables $y$ in $C$ we have $\var(l)\notin\dep(y)$, universal reduction allows us to derive $C\setminus l$.
\emph{Fork Extension} allows us to split a clause $C_1\vee C_2$ by introducing a fresh variable $y$. 
The dependency set of $y$ is defined as the intersection $\dep(C_1) \cap \dep(C_2)$ and represents that the question whether $C_1$ or $C_2$ satisfies the original clause needs to be resolved based on the information that is available to both of them. 
Fork Extension is usually only applied when $C_1$ and $C_2$ have incomparable dependencies ($\dep(C_1) \nsubseteq \dep(C_2)$ and $\dep(C_1) \nsupseteq \dep(C_2)$), as only then the dependency set of $y$ is smaller than those of $C_1$ and of $C_2$. 
We state the rule formally in Fig.~\ref{fig:split-rules}. 

\begin{figure}[t]
  \centering
  $\infer[\forkext]
  {
    \exists y (\dep(C_1) \cap \dep(C_2)) \ldot~ C_1 \!\cup\! \set{y} ~\wedge~ C_2 \!\cup\! \set{\overline{y}}
  }
  {
    C_1 \cup C_2
    \qquad
    y \text{ is fresh}
  }$
  \caption{Fork Extension}
  \label{fig:split-rules}
  \vspace{-10pt}
\end{figure}

\begin{example}  
  As an example of applying the Fork Extension, consider the quantifier prefix
  $\forall x_1 x_2 \ldot \exists y_1(x_1) \ldot \exists y_2(x_2)$ and clause $(\overlineindex{x_1} \lor y_1 \lor y_2)$.
  Applying $\forkext$ with the decomposition $C_1 = \set{\overlineindex{x_1}, y_1}$ and $C_2 = \set{y_2}$ results in the clauses $(\overlineindex{x_1} \lor y_1 \lor \bm{y_3})(\bm{\overlineindex{y_3}} \lor y_2)$ where $y_3$ is a fresh existential variable with dependency set $\dep(y_3) = \emptyset$ ($\dep(C_1) = \set{x_1}$ and $\dep(C_2) = \set{x_2}$).
\end{example}
Resolution is refutationally complete for propositional Boolean formulas. 
This means that for every propositional Boolean formula that is equivalent to false we can derive the empty clause using only Resolution.
In the same way, Resolution and Universal Reduction (together they are called $\qres$) are refutationally complete for QBF~\cite{journals/iandc/BuningKF95}.
For DQBF, however, $\qres$ is not sufficient---it was proven to be sound but incomplete~\cite{journals/tcs/BalabanovCJ14}.
\emph{Fork Resolution} addresses this problem by extending $\qres$ by the Fork Extension proof rule~\cite{conf/sat/Rabe17}. 

Unfortunately, the proof of the completeness of Fork Resolution relied on a hidden assumption that we uncovered by implementing and testing the algorithm proposed in this work.
Consider the DQBF with prefix $
  \forall x_1, x_2, x_3 \ldot \exists y_1(x_1,x_2) \ldot \allowbreak \exists y_2(x_2,x_3) \ldot \exists y_3(x_1,x_3)$ and a clause $C = (y_1 \lor y_2 \lor y_3)$.
Formally, $C$ is an \emph{information fork}~\cite{conf/sat/Rabe17}, i.e., it contains variables with incomparable dependencies. 
However, we cannot apply $\forkext$ because any split of the clause into two parts $C=C_1\vee C_2$ satisfies either $\dep(C_1) \subseteq \dep(C_2)$ or $\dep(C_1) \supseteq \dep(C_2)$.
Fork Extension therefore fails its purpose in this case to eliminate all information forks as required by the proof of completeness in~\cite{conf/sat/Rabe17}. 
We say that information forks that Fork Extension cannot split with a literal with smaller dependency set have a \emph{dependency cycle}. It is easy to extend the example above to a formula for which Fork Resolution is incomplete (see Appendix~\ref{sec:incompleteness-fork}).

We see two ways to counter this problem. 
The first is to consider a normal form of DQBF that does not have dependency cycles.
We can restrict to DQBFs where every incomparable pair of dependency sets must have an empty intersection.
This fragment does not admit any dependency cycles and it is $\nexptime$-complete~\cite{dqbf-nexptime} and therefore could be used as a normal form of DQBF.
It also guarantees that every dependency set that gets introduced through Fork Extension maintains this property (in fact, only variables with the empty dependency set can be created).
In this way, the Fork Resolution proof system is indeed strong enough to serve as a proof system for DQBF.
In fact, most applications already fall in this fragment.
The (Boolean) synthesis of invariants, programs, or winning regions of games can all be expressed as the existence of a function $f\colon \mathbb B^m\to\mathbb B^n$ such that for all tuples of inputs $x_1,\dots,x_k\in\mathbb B^m$ some relation $\varphi(x_1,f(x_1),\dots,x_k,f(x_k))$ over function applications of $f$ is satisfied.
By the typical translation into DQBF this results in a formula with pairwise disjoint dependency sets plus the dependency set $\mathcal X$ for the Tseitin variables~\cite{conf/sat/Rabe17}.
In particular, we have never observed a dependency cycle in the available benchmarks.

\begin{figure*}[t]
  \begin{minipage}{\textwidth}
    \centering
        $\infer[\textbf{S} \forkext]
  {
      \exists y ((\dep(C_1) \cap \dep(C_2)) \setminus \dep(C_X)) \ldot~
      C_X \!\cup\! C_1 \!\cup\! \set{y} ~\wedge~ C_X \!\cup\! C_2 \!\cup\! \set{\overline{y}}
  }
  {
      C_1 \cup C_2 \qquad y \text{ is fresh} \qquad C_X \text{ is a set of universal literals} 
  }$
  \end{minipage}
  \caption{Strong Fork Extension}
  \label{fig:local-expansion-rule}
  \vspace{-10pt}
\end{figure*}

The second approach is to avoid this normal form and strengthen Fork Extension in a way that allows us to break dependency cycles. 
The new rule Strong Fork Extension, depicted in Fig.~\ref{fig:local-expansion-rule}, extends Fork Extension by the ability to introduce new universal literals $C_X$ to the two clauses that it produces.
Intuitively, adding the literals $C_X$ restricts the Skolem function of $y$ to the case that all literals in $C_X$ are false. Hence $y$ does not need to explicitly depend on $\dep(C_X)$.
This allows us to remove $\dep(C_X)$ from the dependency set of the freshly introduced variable $y$.

\begin{lemma}
\label{lem:SFESoundness}
    The Strong Fork Extension rule is sound.
\end{lemma}
\begin{proof}
Given any Skolem function for the formula, we make a case split over the assignments to the universals:
If a literal of $C_X$ is true, both produced clauses are true and the rule is trivially sound.
If all literals of $C_X$ are false, the Strong Fork Extension is equivalent to Fork Extension~\cite{conf/sat/Rabe17}.
\end{proof}

\begin{theorem}  \label{thm:soundness-fork-local-expansion}
  Strong Fork Resolution is sound and complete for DQBF.
\end{theorem}

\begin{proof}
    The proof of completeness of Fork Resolution assumed that any information fork can be split with Fork Extension, by introducing new literals with \emph{smaller} dependency sets~\cite{conf/sat/Rabe17}.
    Strong Fork Extension guarantees this property also for dependency cycles: we pick some universal variable $x$ of the dependency sets and split the original clause twice; once with $C_X=\{x\}$ and once with $C_X=\{\overline x\}$.
    This results in four clauses that together imply the original clause (such that the original clause can be dropped from the formula), and the two variables introduced have smaller dependency sets.
    The rest of the proof remains the same.
\end{proof}

\subsection{Incompleteness Example}  \label{sec:incompleteness-fork}

We give an example to demonstrate that Fork Resolution is incomplete for general DQBF.
The example formula is an extension of the incompleteness examples used in~\cite{journals/fmsd/BalabanovJ12}. 
On the high level, the formula expresses the following:
\[
\begin{array}{l}
  \forall x_1, x_2, x_3 \ldot
  \exists y_1(x_1, x_2) \ldot
  \exists y_2(x_2, x_3) \ldot
  \exists y_3(x_1, x_3) \ldot\qquad\hfill
  \\
  \hfill (x_1 \land x_2 \land x_3) \leftrightarrow (y_1 \oplus y_2 \oplus y_3)    
\end{array}
\]

In CNF, the formula looks as follows:
{
\begin{align*}
&(x_1 \lor x_1 \lor y_2 \lor \overline y_3) \land
(x_1 \lor x_1 \lor y_3 \lor \overline y_2) \land\\
&(x_1 \lor y_2 \lor y_3 \lor \overline x_1) \land
(x_2 \lor x_1 \lor y_2 \lor \overline y_3) \land\\
&(x_2 \lor x_1 \lor y_3 \lor \overline y_2) \land
(x_2 \lor y_2 \lor y_3 \lor \overline x_1) \land\\
&(x_3 \lor x_1 \lor y_2 \lor \overline y_3) \land
(x_3 \lor x_1 \lor y_3 \lor \overline y_2) \land\\
&(x_3 \lor y_2 \lor y_3 \lor \overline x_1) \land
(x_1 \lor \overline x_1 \lor \overline y_2 \lor \overline y_3) \land\\
&(x_2 \lor \overline x_1 \lor \overline y_2 \lor \overline y_3) \land
(x_3 \lor \overline x_1 \lor \overline y_2 \lor \overline y_3) \land\\
&(x_1 \lor y_2 \lor y_3 \lor \overline x_1 \lor \overline x_2 \lor \overline x_3) \land
(x_1 \lor \overline x_1 \lor \overline x_2 \lor \overline x_3 \lor \overline y_2 \lor \overline y_3) \land\\
&(y_2 \lor \overline x_1 \lor \overline x_2 \lor \overline x_3 \lor \overline x_1 \lor \overline y_3) \land
(y_3 \lor \overline x_1 \lor \overline x_2 \lor \overline x_3 \lor \overline x_1 \lor \overline y_2)
\end{align*}
}

\noindent
Note that:
\begin{itemize}
  \item The formula is false.
  \item Universal reduction cannot be applied to any clause.
  \item All resolvents are tautologies.
  \item Fork Extension is not applicable.
  \item The formula does not contain the empty clause.
\end{itemize}
This means that Fork Resolution proof system is not strong enough to refute any DQBF. 
However, we want to emphasize that for the fragment of DQBF that admits only ordered or disjoint dependency sets, which is also $\nexptime$-complete, Fork Resolution is sound and complete, as we discussed in the proof system section in the paper.

The problem in the proof of completeness in~\cite{conf/sat/Rabe17} is that two conflicting definitions of information forks were given. 
In the introduction information forks are defined as clauses that contain two variables with incomparable dependencies (as it is used in this work).
In Section 4 of~\cite{conf/sat/Rabe17} information forks were then defined again as clauses that consist of two parts $C_1$ and $C_2$ that have incomparable dependencies.
The two definitions do not match for clauses that contain three or more variables with pairwise intersecting dependency sets.
This led to the wrong assumption that all information forks (of the first kind) can be eliminated with the Fork Extension rule, which is not the case.

\subsection{Dependency Cycles}

In the following, we formalize dependency cycles and show that they are the reason for incompleteness of Fork Extension.
A clause $C$ contains a \emph{dependency cycle} of length $k > 2$, if there is a subset $\set{l_1, l_2, \dots, l_k} \subseteq C$ of existential literals such that the intersections of dependencies $I_i = \dep(l_i) \cap \dep(l_{i+1}) \neq \emptyset$ for all $1 \leq i \leq k$, with $l_{k+1} = l_1$ contain pairwise disjoint variables, i.e., $I_i \nsubseteq I_j$ and $I_j \nsubseteq I_i$ for each $i \neq j$.
Figure~\ref{fig:dependency-cycle} depicts a representation of the dependency cycle.
Given a clause $C$, the \emph{clause poset}, written $\poset(C)$, is a partially ordered set $\tuple{P, \subseteq}$ where
$P \subseteq \mathcal{X}$ is the set of dependencies of existential literals in $C$, i.e., $\set{ \dep(y) \mid l \in C \land l \text{ is existential} }$.
If $\poset(C)$ contains more than one maximal element w.r.t.~$\subseteq$, $C$ contains a \emph{information fork}.

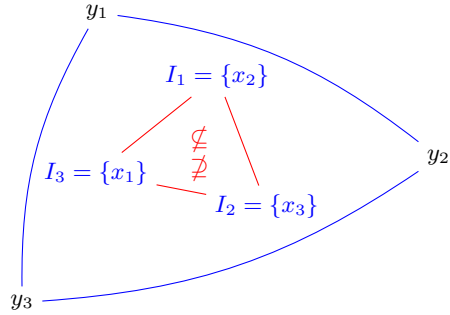
\begin{figure}[t]
  \centering
  \begin{tikzpicture}[auto]
    \node[] (y1) {$y_1$};
    \node[below right=1.5 and 4 of y1] (y2) {$y_2$};
    \node[below left =1.5 and 5 of y2] (y3) {$y_3$};
    
    \draw[blue] (y1) edge[bend left=15] node[swap] (I1) {$I_1 = \set{x_2}$} (y2)
          (y2) edge[bend left=15] node[swap,near start] (I2) {$I_2 = \set{x_3}$} (y3)
          (y3) edge[bend left=15] node[swap] (I3) {$I_3 = \set{x_1}$} (y1);
    \draw[red] (I1) edge (I2)
          (I2) edge node[swap,align=center] {$\nsubseteq$\\$\nsupseteq$} (I3)
          (I3) edge (I1);
  \end{tikzpicture}  
  \caption{Visualization of the dependency cycle for clause $(y_1 \lor y_2 \lor y_3)$ with prefix $\forall x_1, x_2, x_3 \ldot \exists y_1(x_1,x_2) \ldot \allowbreak \exists y_2(x_2,x_3) \ldot \exists y_3(x_1,x_3)$.}
  \label{fig:dependency-cycle}
\end{figure}

\begin{lemma}  \label{thm:dependency-fork-cycle}
  Fork Extension is applicable for clauses with information fork if, and only if, the clause does not contain a dependency cycle.
\end{lemma}
\begin{proof}
Assume $C$ contains a dependency cycle, that is, there is a $k > 2$ and $\set{l_1, l_2, \dots, l_k} \subseteq C$ of existential literals such that $I_i = \dep(l_i) \cap \dep(l_{i+1}) \neq \emptyset$ for all $1 \leq i \leq k$, with $l_{k+1} = l_1$ contain pairwise disjoint variables, i.e., $I_i \nsubseteq I_j$ and $I_j \nsubseteq I_i$ for each $i \neq j$.
W.l.o.g.~we assume that $\set{l_1, l_2, \dots, l_k}$ are the only existential variables in $C$.
Let $C_1 \cup C_2$ be an arbitrary split containing at least one existential variable.
Then, either $\forkext$ is not applicable ($\dep(C_1) \subseteq \dep(C_2)$ or $\dep(C_1) \supseteq \dep(C_2)$) or applying it will lead to a clause with dependency cycle: let $y$ be the fresh variable with $\dep(y) = \dep(C_1) \cap \dep(C_2)$.
If $C_1$ contains a single existential literal $l_i$, then $\dep(y) \cap \dep(l_{i+1}) \neq \emptyset$ and $\dep(y) \cap \dep(l_{i-1}) \neq \emptyset$, i.e., the resulting clause has a dependency cycle of length $k$.
If $C_1$ contains $j > 1$ existential literals, then both resulting clauses contain a dependency cycle of length $j + 1$ and $k - j + 1$.

Assume $C$ does not contain a dependency cycle, that is, there is a maximal element $H$ in $\poset(C)$, such that there is a unique maximal element in the set of intersections with other maximal elements $H^* = \max_\subseteq\set{H \cap H' \mid H' \text{ is a maximal element of }\poset(C)}$.
We use the Fork Extension rule $\forkext$ with $C_1 = \set{l \in C \mid \dep(l) \subseteq H, \dep(l) \not\subseteq H^*}$ and $C_2 = C \setminus C_1$.
\end{proof}

A DQBF formula is in the \emph{multi-linear} fragment of DQBF, if for all pairs of existential variables $y_1$ and $y_2$ with dependency sets $H_1$ and $H_2$ it holds that $H_1 \subseteq H_2$, $H_2 \subseteq H_1$, or $H_1 \cap H_2 = \emptyset$.
\begin{theorem}
  The multi-linear fragment of DQBF does not contain dependency cycles.
\end{theorem}
\begin{proof}
  Assume we have three variables $y_1$, $y_2$, and $y_3$ with dependency sets $H_1$, $H_2$, and $H_3$.
  Further, let $H_1 \cap H_2 \neq \emptyset$ and $H_2 \cap H_3 \neq \emptyset$, that is, $H_1 \subseteq H_2$ or $H_2 \subseteq H_1$, and $H_2 \subseteq H_3$ or $H_3 \subseteq H_2$, thus, there are 4 combinations:
  \begin{itemize}
    \item $H_1 \subseteq H_2$ and $H_2 \subseteq H_3$, thus $H_1 \subseteq H_3$ and $H_1 \cap H_2 \subseteq H_2 \cap H_3$
    \item $H_1 \subseteq H_2$ and $H_3 \subseteq H_2$, thus $H_1 \cap H_3 \subseteq H_1 \cap H_2$
    \item $H_2 \subseteq H_1$ and $H_2 \subseteq H_3$, thus $H_1 \cap H_2 = H_2 \cap H_3$
    \item $H_2 \subseteq H_1$ and $H_3 \subseteq H_2$, thus $H_3 \subseteq H_1$ and $H_2 \cap H_3 \subseteq H_1 \cap H_2$
  \end{itemize}
  which rules out any dependency cycle.
\end{proof}
\begin{corollary}
  Fork Extension is complete for the multi-linear fragment of DQBF.
\end{corollary}

\section{Correctness}
\label{sec:correctness-proof}

In this section, we give a formal correctness proof of the algorithm.
We start by giving the correctness arguments for the base case and state theorems over the structure of the abstractions.
Then, we split the actual correctness proof into two theorems that argue inductively over the structure of the quantifier levels.

The first lemma states the base case, i.e., that the abstraction for the maximal element is equisatisfiable to replacing the assignment of the bound variables $\assignment_V$ in the matrix $\varphi$.
\begin{lemma}  \label{thm:correctness-base-case}
  Let $\tuple{\exists, Y, H}$ be the existential node corresponding to the unique maximal element and let $\assignment_V$ be some assignment with $V = \bound(\tuple{\exists, Y, H})$.
  Then, the SAT call (line~\ref{line:callSATexistentialnode}) of \Call{solve$_\exists$}{$\tuple{\exists, Y, H}$} returns $\mathsf{Sat}$ if, and only if, $\varphi[\assignment_V]$ is satisfiable.
\end{lemma}
\begin{proof}
  The abstraction $\theta_Y$ for the maximal element does not contain assumption literals, i.e., it has the form $\theta_Y = \bigwedge_{C_i \in \matrix} s_i \lor C_i^=$.
  By definition of $\assignment_S = \prj_\exists(Y, \assignment_V)$, it holds that $\theta_Y[\assignment_S] = \bigwedge\limits_{C_i \in \matrix \atop \assignment_V \nmodels C_i^<} C_i^= = \varphi[\assignment_V]$.
\end{proof}
Additionally, we state the following Lemma for non-maximal nodes.
\begin{lemma}  \label{thm:existential-abstraction}
  Let $\tuple{\exists, Y, H}$ be an existential node and let $\assignment_S$ be some assignment of the satisfaction variables.
  It holds that $\theta_Y[\assignment_S] = \bigwedge_{C_i \in \matrix, \assignment_S(s_i) = \bot} (C_i^= \lor a_i)$.
\end{lemma}
\begin{proof}
  The abstraction $\theta_Y$ for the an existential node $\tuple{\exists, Y, H}$ has the form $\theta_Y = \bigwedge_{C_i \in \matrix} ( a_i \lor s_i \lor C_i^=)$.
  It follows immediately that $\theta_Y[\assignment_S] = \bigwedge_{C_i \in \matrix, \assignment_S(s_i) = \bot} (C_i^= \lor a_i)$.
\end{proof}
\begin{lemma}  \label{thm:universal-abstraction}
  Let $\tuple{\forall, X}$ be an universal node and let $\assignment_S$ be some positive assignment of the satisfaction variables (i.e., a partial assignment containing only positive values).
  It holds that $\theta_X[\assignment_S] = \bigwedge_{C_i \in \matrix, \assignment_S(s_i) \neq \top} (s_i \lor \neg C_i^=)$.
\end{lemma}
\begin{proof}
  The abstraction $\theta_X$ for the a universal node $\tuple{\forall, X}$ has the form $\theta_X = \bigwedge_{C_i \in \matrix} ( s_i \lor  \neg C_i^=)$.
  It follows immediately that $\theta_X[\assignment_S] = \bigwedge_{C_i \in \matrix, \assignment_S(s_i) \neq \top} (s_i \lor \neg C_i^=)$.
\end{proof}
The following lemmata state that refinements are correct, i.e., that the clause contained in the refinement is satisfied, respectively, falsified.
\begin{lemma}  \label{thm:refinement-unsat}
  Let $\tuple{\exists, Y, H}$ be some existential node and let $\assignment_V$ be some assignment with $V = \bound(\tuple{\exists, Y, H})$. Let $\assignment$ be the assignment after a satisfiable call to the abstraction $\theta_Y$ (line~\ref{line:callSATexistentialnode} of \Call{solve$_\exists$}{$\tuple{\exists, Y, H}$}).
  For every clause $C_i \in \matrix$ it holds that $a_i \mapsto \bot$ implies that $\assignment_Y \mathbin{\dot\cup} \assignment_V \models C_i$.
\end{lemma}
\begin{proof}
  Follows by the abstraction definitions and the projection functions.
\end{proof}
\begin{lemma}  \label{thm:refinement-sat}
  Let $\tuple{\forall, X}$ be some universal node and let $\assignment_V$ be some assignment with $V = \bound(\tuple{\forall, X})$. Let $\assignment$ be the assignment after a satisfiable call to the abstraction $\theta_X$ (line~\ref{line:callSATuniversalnode} of \Call{solve$_\forall$}{$\tuple{\forall, X}$}).
  For every clause $C_i \in \matrix$ it holds that $s_i \mapsto \bot$ implies that $\assignment_X \mathbin{\dot\cup} \assignment_V \nmodels C_i$.
\end{lemma}
\begin{proof}
  Follows by the abstraction definitions and the projection functions.
\end{proof}

The proof of correctness is an inductive argument over the quantifier levels.
Fix some level $i$ and an assignment of the variables bound before $i$, the algorithm determines the result of the DQBF where the prior bound variables are replaced by the assignment.
The algorithm, further, determines a subset of the satisfied clauses as a witness for the outer levels.

We define an operator $\Phi \downarrow_\mathit{lvl} \assignment_S$ that restricts the matrix $\matrix$ in a DQBF $\Phi$ to those clauses $C_i \in \matrix$ such that $\assignment_S(s_i) = \bfalse$, i.e., the resulting DQBF has the same quantifier prefix from quantifier level $\mathit{lvl}$ onwards with matrix $\matrix' \coloneqq \set{C_i^\geq \mid C_i \in \matrix \land \assignment_S(s_i) = \bfalse}$.
Variables that are bound by a smaller quantifier level than $\mathit{lvl}$ are removed from the matrix.
Intuitively, the operator removes clauses marked as satisfied by $\assignment_S$.

For a partial assignment $\assignment$, we use the notation $\assignment[\bundef \mapsto b]$ to denote the complete assignment where undefined values are replaced by $b \in \set{\btrue, \bfalse}$.

\begin{lemma}  \label{thm:dqbf-induction-sat}
  Let $\Phi$ be a DQBF with matrix $\matrix$, let $\mathit{lvl}$ be a quantifier level, and let $\assignment_V$ be an assignment of variables bound prior to $\mathit{lvl}$.
  If $\Phi[\assignment_V]$ is true \Call{SolveLevel}{$\mathit{lvl}$} produces a sat conflict with partial assignment $\assignment_S$ such that $\Phi \downarrow_\mathit{lvl} \assignment_S[\bundef \mapsto \bfalse]$ is true.
\end{lemma}
\begin{proof}
  We prove the statement by induction over the quantifier levels.
  
  Let $\mathit{lvl}$ be the quantifier level with the unique maximal node $N_\mathit{max} = \tuple{\exists, Y, H}$ (see~\autoref{thm:unique-maximal}) and let $\assignment_V$ be such that $\Phi[\assignment_V]$ is true.
  By \autoref{thm:correctness-base-case}, the truth of $\Phi[\assignment_V]$ witnesses the satisfiability of $\theta_Y[\assignment_S]$ where $\assignment_S = \prj(Y, \assignment_V)$.
  As $N_\mathit{max}$ is maximal, the algorithm \Call{Solve$_\exists$}{} calls \Call{refine}{$\sat$, $\assignment_S$, $N_\mathit{max}$} and $\assignment_S[\bundef \mapsto \bfalse]$ is equivalent to $\assignment_S$ satisfying the second condition due to the definition of $\prj_\exists$.
  
  Let $\mathit{lvl}$ be an existential quantifier level and let $\assignment_V$ be such that $\Phi[\assignment_V]$ is true.
  Let $\tuple{\exists, Y, H}$ be an arbitrary existential node in $\mathit{lvl}$.
  Further, let $\assignment_S = \prj_\exists(Y, \assignment_V)$.
  By \autoref{thm:existential-abstraction} it holds that 
  \begin{equation*}
    \theta_Y[\assignment_S] = \bigwedge_{C_i \in \matrix \mid \assignment_S(s_i) = \bfalse} \left( C_i^= \lor a_i \right)  \enspace.
  \end{equation*}
  Since $\Phi[\assignment_V]$ and thereby $\Phi \downarrow_\mathit{lvl} \assignment_S$ is true, there is a satisfying assignment $\assignment_Y$ for the variables $Y$ such that $(\Phi \downarrow_\mathit{lvl} \assignment_S)[\assignment_Y]$ is true.
  Define $\assignment^*_A$ as $\assignment^*_A(a_i) = \bfalse$ if, and only if, $\assignment_V \asscupdot \assignment_Y \models C_i^\leq$.
  Thus, $\assignment^*_A$ is the minimal  assignment with respect to the number of assumptions ($\assignment^*_A(a_i) = \btrue$) for the given assignment $\assignment_Y$.
  The combined assignment $\assignment_X \asscupdot \assignment^*_A$ is a satisfying assignment of the initial abstraction $\theta_Y[\assignment_S]$ by construction.
  Thus, for every node $N$ in $\mathit{lvl}$, \Call{solve$_\exists$}{} returns $\mathsf{CandidateFound}$ and the algorithm continues to the next quantifier level.
  We do a case distinction on the assignments created on $\mathit{lvl}$, i.e., returned by the SAT solver in line~\ref{line:update-assignments-existential}.
  As $\Phi[\assignment_V]$ is true, there is a satisfying assignment $\assignment_Y^*$ for every node $\tuple{\exists, Y, H}$ in $\mathit{lvl}$.
  
  Assume that the SAT solver in line~\ref{line:update-assignments-existential} returns this assignment.
  Thus, $\Phi[\assignment_V \asscupdot \assignment_{Y_1}^* \cdots \asscupdot \assignment_{Y_n}^*]$ is true.
  By induction hypothesis we deduce that the next level produces a sat conflict with partial assignment $\assignment_S$ such that $\Phi \downarrow_{\mathit{lvl}+1} \assignment_S[\bundef \mapsto \bfalse]$ is true, i.e., the assignment $\assignment_S$ represents those clauses that need to be satisfied such that $\Phi$ is true.
  Since $\mathit{lvl}$ is existential, this witness is propagated. \thesistodo{need to adapt $\assignment_S$ based on the assignments $\assignment_Y$}
  
  Assume that the SAT solver in line~\ref{line:update-assignments-existential} returns a different assignment.
  If the assignment is still satisfying $\Phi$, the next level returns a sat conflict and the same argumentation as above applies.
  In the case the next level returns a unsat conflict with witness $\assignment'_S$ there are 3 possibilities:
  \begin{enumerate}
    \item The conflict does not contain variables of any existential node, which immediately contradicts that $\Phi[\assignment_V]$ is true.
    \item The conflict contains variables of a single existential node $\tuple{\exists, Y, H}$.
      The subsequent refinement in line~\ref{line:refine-unsat} of Fig.~\ref{alg:refinement} requires that one of the not satisfied clauses $C_i$ with $\assignment'_S(s_i) = \bfalse$ has to be satisfied in the next iteration and the corresponding refinement clause is $\psi \coloneqq \bigvee_{C_i \in \matrix \mid \assignment'_S(s_i) = \bfalse} \overlineindex{a_i}$.
      By construction of $\assignment^*_A$ as the minimal assignment corresponding to $\assignment_Y$, $\assignment^*_A \nmodels \psi$ contradicts that $\assignment_Y$ is a satisfying assignment of $\Phi[\assignment_V]$.
      Hence, $\assignment_Y \asscupdot \assignment^*_A$ is still a satisfying assignment for the refined abstraction $\theta'_Y[\assignment_S]$.
      The refinement also reduces the number of $A$ assignments by at least 1 and, thus, brings us one step closer to termination.
    \item The conflict contains variables of more than one existential node, thus, the conflict clause $C_\mathit{conflict}$ in line~\ref{line:build-conflict-clause} of Fig.~\ref{alg:refinement} contains an information fork.
      It holds that $\assignment_V \asscupdot \assignment_Y^\cup \nmodels C_\mathit{conflict}$~\cite{conf/cav/Tentrup17}, where $\assignment_Y^\cup$ is the combined assignment of $\mathit{lvl}$.
      Applying Fork Extension gives us new clauses without information fork, and the new DQBF $\Phi[\assignment'_V]$ is still true (where $\assignment'_V$ is the assignment $\assignment_V$ plus added assignment for the new variables due to Fork Extension).
      Unlike before, $\assignment_Y^\cup$ does no longer satisfy the abstraction, thus, a different assignment is produced.
  \end{enumerate}
  In all possible cases, eventually, the satisfying assignment is reached.
  
  Let $\mathit{lvl}$ be a universal quantifier level with the singleton node $\tuple{\forall, Y}$ and let $\assignment_V$ be such that $\Phi[\assignment_V]$ is true.
  Further, let $\assignment_S = \prj_\forall(X, \assignment_V)$.
  For every assignment $\assignment_X$, it holds that $\Phi[\assignment_V \asscupdot \assignment_X]$ is true.
  By \autoref{thm:universal-abstraction} it holds that 
  \begin{equation*}
    \theta_X[\assignment_S] = \bigwedge_{C_i \in \matrix, \assignment_S(s_i) \neq \top} (s_i \lor \neg C_i^=)  \enspace.
  \end{equation*}
  Thus, in order to set $s_i$ to false for some $i$, every literal $l \in C_i^=$ has to be assigned negatively.
  Fix some arbitrary assignment $\assignment_X$.
  By induction hypothesis, the following level produces a sat conflict with partial assignment $\assignment'_S$ such that $\Phi \downarrow_{\mathit{lvl}+1} \assignment'_S[\bundef \mapsto \bfalse]$ is true.
  The subsequent refinement in line~\ref{line:refine-sat} of Fig.~\ref{alg:refinement} reduces the number of $S$ assignments, thus, eventually, the abstraction $\theta_X[\assignment_S]$ becomes unsatisfiable.
  Let $\theta'_X$ be the abstraction at this point and let $\assignment^*_S$ be the failed assumptions.
  $\assignment^*_S \subassign \assignment_S^+$ holds as $\assignment^*_S$ are the failed assumptions of the SAT call \Call{sat}{$\theta'_X$, $\assignment_S$}.
  
  It remains to show that $\Phi[\assignment_V] \downarrow_\mathit{lvl} \assignment^*_S[\bundef \mapsto \bfalse]$ is true.
  Assume for contradiction that there is some $\assignment_X$ such that $(\Phi[\assignment_V] \downarrow_\mathit{lvl} \assignment^*_S[\bundef \mapsto \bfalse])[\assignment_X]$ is false.
  Let $\assignment_S = \prj_\forall(X, \assignment_V)$
  We know that $\theta'_X[\assignment_X \asscupdot \assignment_S]$ is unsatisfiable.
  Thus, the assignment $\assignment_X$ was excluded due to refinements.
  As the refinement only excludes $S$ assignments $\assignment''_S$ such that $\Phi \downarrow_{\mathit{lvl}+1} \assignment''_S[\bundef \mapsto \bfalse]$ is true, this leads to a contradiction.
\end{proof}

\begin{lemma}  \label{thm:dqbf-induction-unsat}
  Let $\Phi$ be a DQBF with matrix $\matrix$, let $\mathit{lvl}$ be a quantifier level, and let $\assignment_V$ be an assignment of variables bound prior to $\mathit{lvl}$.
  If $\Phi[\assignment_V]$ is false \Call{SolveLevel}{$\mathit{lvl}$} produces a unsat conflict with partial assignment $\assignment_S$ such that $\Phi \downarrow_\mathit{lvl} \assignment_S[\bundef \mapsto \btrue]$ is false.
\end{lemma}
\begin{proof}
    We prove the statement by induction over the quantifier levels.
  
  Let $\mathit{lvl}$ be the quantifier level with the unique maximal node $N_\mathit{max} = \tuple{\exists, Y, H}$ (see~\autoref{thm:quantifier-levels}.\ref{thm:unique-maximal}) and let $\assignment_V$ be such that $\Phi[\assignment_V]$ is false.
  By \autoref{thm:correctness-base-case}, $\theta_Y[\assignment_S]$ is unsatisfiable where $\assignment_S = \prj(Y, \assignment_V)$.
  Let $\assignment'_S$ be the failed assumptions from the sat call to \Call{sat}{$\theta_Y, \assignment_S}$, i.e., $\assignment'_S \subassign \assignment_S$ and $\theta_Y[\assignment'_S]$ is unsatisfiable.
  Due to the definition of the abstraction, $\Phi \downarrow_\mathit{lvl} \assignment'_S[\bundef \mapsto \btrue]$ is false.
  
  Let $\mathit{lvl}$ be an existential quantifier level and let $\assignment_V$ be such that $\Phi[\assignment_V]$ is false.
  Let $\tuple{\exists, Y, H}$ be an arbitrary existential node in $\mathit{lvl}$.
  Further, let $\assignment_S = \prj_\exists(Y, \assignment_V)$.
  By \autoref{thm:existential-abstraction} it holds that 
  \begin{equation*}
    \theta_Y[\assignment_S] = \bigwedge_{C_i \in \matrix \mid \assignment_S(s_i) = \bfalse} \left( C_i^= \lor a_i \right)  \enspace.
  \end{equation*}
  As $\Phi[\assignment_V]$ is false, every assignment of the existential level is false as well.
  There are two possible executions: 
  \begin{itemize}
    \item Assume that all existential nodes generate a candidate assignment, then we can apply the induction hypothesis to deduce that the next level produces a unsat conflict with partial assignment $\assignment'_S$ such that $\Phi \downarrow_{\mathit{lvl}+1} \assignment'_S[\bundef \mapsto \btrue]$ is false.
      The refinement with witness $\assignment'_S$ has three possibilities:
      \begin{enumerate}
        \item The conflict does not contain variables of any existential node, that is, the algorithm produces the partial assignment $\assignment'_S$.
        \item The conflict contains variables of a single existential node $\tuple{\exists, Y, H}$.
          The subsequent refinement in line~\ref{line:refine-unsat} of Fig.~\ref{alg:refinement} requires that one of the not satisfied clauses $C_i$ with $\assignment'_S(s_i) = \bfalse$ has to be satisfied in the next iteration and the corresponding refinement clause is $\psi \coloneqq \bigvee_{C_i \in \matrix \mid \assignment'_S(s_i) = \bfalse} \overlineindex{a_i}$.
          The refinement reduces the number of $A$ assignments by at least 1 and, thus, brings us one step closer to termination.
        \item The conflict contains variables of more than one existential node, thus, the conflict clause $C_\mathit{conflict}$ in line~\ref{line:build-conflict-clause} of Fig.~\ref{alg:refinement} contains an information fork.
          It holds that $\assignment_V \asscupdot \assignment_Y^\cup \nmodels C_\mathit{conflict}$~\cite{conf/cav/Tentrup17}, where $\assignment_Y^\cup$ is the combined assignment of $\mathit{lvl}$.
          Applying Fork Extension gives us new clauses without information fork. In the following, a different assignment is produced.
      \end{enumerate}
      In the latter two cases, we make progress towards termination.
    \item Assume that one of the existential nodes $\tuple{\exists, Y, H}$ produce a unsat conflict.
      Let $\theta'_Y$ be the abstraction at this point and let $\assignment'_S$ be the failed assumptions, i.e., $\assignment'_S \subassign \assignment_S$.
      
      Let $\assignment''_S = \assignment'_S[\bundef \mapsto \btrue]$.
      It remains to show that $\Phi \downarrow_\mathit{lvl} \assignment''_S$ is false.
      Assume for contradiction that there is some $\assignment_Y$ such that $(\Phi \downarrow_\mathit{lvl} \assignment''_S)[\assignment_Y]$ is true.
      It holds that $\theta'_Y[\assignment_Y \asscupdot \assignment''_S]$ is unsatisfiable, whereas initially, $\theta_Y[\assignment_Y \asscupdot \assignment''_S]$ is satisfiable.
      Thus, the assignment $\assignment_Y$ was excluded due to refinements.
      As the refinement only excludes assignments corresponding to some $S$ assignment $\assignment^*_S$ such that $\Phi \downarrow_\mathit{lvl} \assignment^*_S[\bundef \mapsto \btrue]$ is false, this contradicts our assumption.
  \end{itemize}
  
  Let $\mathit{lvl}$ be a universal quantifier level with the singleton node $\tuple{\forall, Y}$ and let $\assignment_V$ be such that $\Phi[\assignment_V]$ is false.
  Further, let $\assignment_S = \prj_\forall(X, \assignment_V)$.
  There is some assignment $\assignment_X$ such that that $\Phi[\assignment_V \asscupdot \assignment_X]$ is false.
  By \autoref{thm:universal-abstraction} it holds that 
  \begin{equation*}
    \theta_X[\assignment_S] = \bigwedge_{C_i \in \matrix, \assignment_S(s_i) \neq \top} (s_i \lor \neg C_i^=)  \enspace.
  \end{equation*}
  $\theta_X[\assignment_S]$ is initially satisfiable by construction.
  Given $\assignment_X$ from , we define the optimal corresponding assignment $\assignment^*_S$ as $\assignment^*_S(s_i) = \btrue$ if, and only if, either $\assignment_S(s_i) = \btrue$ or $\assignment_X \models C_i^=$.
  Assume that the SAT solver in line~\ref{line:callSATuniversalnode} of Fig.~\ref{alg:node-solving} returns the assignment $\assignment_X$.
  Thus, by induction hypothesis, the next level produces a unsat conflict with partial assignment $\assignment'_S$ such that $\Phi \downarrow_{\mathit{lvl}+1} \assignment'_S[\bundef \mapsto \btrue]$ is false.
  
  Assume that the SAT solver in line~\ref{line:callSATuniversalnode} of Fig.~\ref{alg:node-solving} returns a different assignment $\assignment'_X$.
  If $\Phi[\assignment_V \asscupdot \assignment'_X]$ is false, the same argumentation as above applies.
  If this is not the case, the next level produces a sat conflict with partial assignment $\assignment'_S$ such that $\Phi \downarrow_{\mathit{lvl}+1} \assignment'_S[\bundef \mapsto \bfalse]$ is true. 
  Subsequently, $\theta_X$ is refined by adding the the clause $\psi \coloneqq \bigvee_{C_i \in \matrix \mid \assignment'_S(s_i) = \btrue} \overline{s_i}$.
  By construction of $\assignment^*_S$ as the optimal assignment corresponding to $\assignment_X$, we deduce that $\assignment^*_S \nmodels \psi$ contradicts that $\assignment_X$ is a witness that $\Phi[\assignment]$ is false.
  Thus, $\assignment_X \asscupdot \assignment^*_S$ remains a satisfying assignment of the refined abstraction.
  The refinement reduced the number of $S$ assignments and, thus, the falsifying assignment $\assignment_X$ is reached eventually.
\end{proof}

\begin{theorem}
  \Call{Solve}{$\Phi$} returns $\sat$ if, and only if, $\Phi$ is satisfiable.
\end{theorem}

\end{document}